\documentclass[preprint,12pt]{elsarticle}

\usepackage{amssymb}
\usepackage{amsmath}
\usepackage{amsthm}
\usepackage{xcolor}
\usepackage{hyperref}
\hypersetup{
    colorlinks,
    linkcolor={red!50!black},
    citecolor={blue!50!black},
    urlcolor={blue!80!black}
}

\journal{International Journal of Engineering Science}

\usepackage{bm}
\usepackage{hyperref}
\usepackage{soul}
\usepackage[nameinlink]{cleveref}
\usepackage{mathtools}
\usepackage{todonotes}
\newcommand{\abs}[1]{\left\lvert#1\right\rvert}
\let\vec\bm
\newtheorem{remark}{Remark}
\newtheorem{definition}{Definition}
\newtheorem{proposition}{Proposition}
\newtheorem{principle}{Principle}
\newtheorem{theorem}{Theorem}

\newtheorem{example}{Example}
\newtheorem{criterion}{Selection Procedure}

\newcommand{\uze}[1]{#1}
\newcommand{\rev}[1]{#1}
\usepackage{comment}

\begin{document}

\begin{frontmatter}



\title{A kinetic interpretation of thermomechanical restrictions of continua\footnote{In memory of Professor K.~R.~Rajagopal,  
\textborn~November 24, 1950 -- \textdied~March 20, 2025
}}


\author[OX,Prague]{Patrick E.~Farrell} 
\author[Prague]{Josef M\'{a}lek} 
\author[Prague]{Ond\v{r}ej Sou\v{c}ek} 
\author[OX]{Umberto Zerbinati} 

\affiliation[OX]{organization={Mathematical Institute, University of Oxford},
            addressline={Andrew Wiles Building, Woodstock Road},
            city={Oxford},
            postcode={OX2 6GG},
	    state={Oxfordshire},
	    country={United Kingdom}
}
\affiliation[Prague]{organization={Mathematical Institute, Faculty of Mathematics and Physics, Charles University},
	    addressline={Sokolovská 49/83},
            city={Prague},
            postcode={186 75},
	    state={Prague},
	    country={Czechia}
}


\begin{abstract}
Rajagopal and Srinivasa's thermodynamic framework derives constitutive relations in continuum mechanics from two scalar functions describing energy storage and entropy production via a constrained optimization principle. In parallel, kinetic theory obtains constitutive laws through moment closure, most notably via the Chapman--Enskog expansion.

\rev{
	This work has three objectives.
    First, we establish a connection between these approaches by providing a kinetic interpretation of the Rajagopal--Srinivasa principle of maximal entropy production, under appropriate albeit restrictive hypotheses.
    For a Bhatnagar--Gross--Krook-type approximation, we show that the Rajagopal--Srinivasa principle is equivalent to a minimal relaxation-time principle, selecting among admissible constitutive responses the one with the fastest compatible relaxation toward equilibrium.
    Second, we review the classical kinetic description of continua in a manner accessible to those familiar with continuum thermodynamics.}

	\rev{Third, we propose a hybrid Chapman--Enskog--Rajagopal--Srinivasa approach which computes the thermodynamic relations and entropy production from the Chapman--Enskog expansion, and then invokes the Rajagopal--Srinivasa principle to determine the other constitutive relations. This recovers the standard Euler and Navier--Stokes--Fourier constitutive laws for monatomic gases.
	We also demonstrate how different choices of selection procedure can be more informative than the classical Chapman--Enskog closure in the context of an inviscid compressible Leslie--Ericksen model arising in liquid crystals.
}
\end{abstract}


\end{frontmatter}




\section{Introduction}
The construction of constitutive relations that are both physically meaningful and thermodynamically consistent is a central problem in continuum mechanics.
In a seminal contribution, Rajagopal and Srinivasa \cite{rajagopalSrinivasa} showed that material behavior can be fully characterized by two scalar-valued functions describing energy storage and entropy production.
By formulating an appropriate constrained optimization problem, all constitutive relations—including vectorial and tensorial ones—follow automatically, while consistency with the second law of thermodynamics is ensured \cite{vitMalek}.
This framework has since been successfully applied to a wide range of systems, including viscoelastic fluids \cite{Mal-Pru-Skr-Sul18,mkrrkt2015,mkrrkt2018}, phase-field models \cite{chn-heida,achn-heida}, Korteweg-type fluids \cite{heidaMalek, Sou-Hei-Mal20}, and dilute polymeric fluids \cite{dost}.

Independently, kinetic theory provides a mesoscopic route to constitutive modeling.
Starting from the Boltzmann equation, macroscopic balance laws and constitutive relations are obtained via moment closures, most notably through the Chapman--Enskog expansion.
In this setting, the caloric equation of state, entropy balance, and entropy production arise naturally, while explicit constitutive laws—such as those of Euler and Navier--Stokes--Fourier theory—are recovered as successive asymptotic approximations near thermodynamic equilibrium.

Despite their shared thermodynamic foundations, the relationship between these two approaches has remained largely implicit.
The Rajagopal--Srinivasa framework is formulated entirely at the continuum level, whereas kinetic theory derives constitutive laws through increasingly intricate asymptotic expansions.
Beyond first order, classical Chapman--Enskog closures are known to lose thermodynamic consistency by violating the non-negativity of entropy production. 
This raises a natural question: \textit{can kinetic theory be used to inform thermodynamic optimization principles, while avoiding the complex calculations required for higher-order moment closures and retaining thermodynamic consistency by construction?}

\rev{
	The purpose of this work is to make the connection more precise and to partially address this question.
	We do so in two steps.
		We first derive well-known results from kinetic theory and recast them in a framework more familiar to a continuum mechanics audience. This includes the balance equations, the relation between the entropy and energy fluxes, and the constitutive equation for the entropy and the rate of entropy production.
		We then recall the key thermodynamic ingredients underlying the Rajagopal--Srinivasa framework. On this basis, we propose a hybrid Chapman--Enskog--Rajagopal--Srinivasa methodology, in which the Chapman--Enskog expansion is employed only to compute entropy production and equilibrium thermodynamic relations, while the full constitutive closures are obtained through constrained optimization.
	This strategy preserves thermodynamic consistency by construction and avoids the explicit computation of higher-order Chapman--Enskog corrections.
}

A central conceptual contribution of the paper concerns the interpretation of the optimization principle itself. For Bhatnagar--Gross--Krook-type kinetic models, we show that the principle of maximal entropy production can be reinterpreted as a minimal relaxation-time principle: among all admissible constitutive responses compatible with a given order Chapman--Enskog truncation, the physically realized response corresponds to the fastest admissible relaxation toward equilibrium.
This provides a kinetic rationale for the selection mechanism underlying the Rajagopal--Srinivasa framework and offers a new perspective on entropy production maximization.

We illustrate the approach for monatomic gases, recovering the Euler and Navier--Stokes--Fourier constitutive laws.
In these settings, the proposed framework reproduces classical results while significantly reducing the complexity of the derivation and clarifying the thermodynamic structure of the closures.
We also demonstrate, in the context of a compressible Leslie--Ericksen model arising from a different class of kinetic equations, that hybrid approaches between kinetic theory and continuum thermodynamics may provide alternative closures to the classical Chapman--Enskog procedure \cite{farrell}.

\section{The Boltzmann equation}
\label{sec:kinetic_theory_background}

The Boltzmann equation describes the statistical behavior of a system of particles comprising a rarefied gas out of equilibrium.
In particular, it describes the evolution of the one-particle distribution function $f(\vec{x},\vec{v},t)$, describing the density of particles at position $\vec{x}$, with velocity $\vec{v}$, at time $t$, via the partial differential equation
\begin{equation}
\label{eq:boltzmann}
   \partial_t f + \vec{v} \cdot \nabla_{\vec{x}} f = Q(f,f),
\end{equation}
where $Q(f,f)$ is the collision operator, which in the classical Boltzmann equation accounts for the binary collisions between particles \cite{cercignani,Harris}. The collision operator has been generalized by many authors to account for different types of interactions and particle systems \cite{pareschiToscani}. For simplicity we assume here that the equation is posed on the entirety of $(\vec{x},\vec{v},t) \in \mathbb{R}^3 \times \mathbb{R}^3 \times (0, \infty)$, with sufficiently rapid decay of $f$ at infinity to drop boundary terms.

Typically, for molecular interactions with spherical symmetries, these generalized collision operators can be expressed in terms of a scattering kernel $B(\beta, \varphi, \abs{\vec{v}-\vec{v}_*})$, which quantifies the rate of collisions between particles with relative velocity $\abs{\vec{v}-\vec{v}_*}$ and scattering angles $\beta$ and $\varphi$. We thus express the collision operator as 
\begin{equation}
	\label{eq:full_collision}
	Q(f, f) = \int_{\mathbb{R}^3} \int_{\mathbb{S}^2} B(\beta, \varphi, \abs{\vec{v}-\vec{v}_*}) \Big( f'f_*' - ff_* \Big) d\beta\,d\varphi\, d\vec{v}_*,
\end{equation}
where $\mathbb{S}^2$ is the unit sphere in three dimensions, and $\vec{v}'$ and $\vec{v}'_*$ are the post-collision velocities determined uniquely via a prescribed binary collision rule from the pre-collision velocities $\vec{v}$ and $\vec{v}_*$ and the scattering angles. We denote $f'\coloneqq f(\vec{x},\vec{v}',t)$, $f'_*\coloneqq f(\vec{x},\vec{v}'_*,t)$, $f\coloneqq f(\vec{x},\vec{v},t)$, and $f_*\coloneqq f(\vec{x},\vec{v}_*,t)$ for brevity.

The use of the full collision operator \eqref{eq:full_collision} makes calculations quite involved. To keep our exposition simple, we instead employ the
Bhatnagar--Gross--Krook (BGK) collision operator \cite{BGK}, given by
\begin{equation} \label{eq:bgk}
	Q(f,f) \approx Q_{\text{BGK}}(f, f^{(0)}) \coloneqq -\frac{1}{\tau}(f - f^{(0)}),
\end{equation}
where $\tau=\tau(\vec{x},t)$ is the local relaxation time, and $f^{(0)}$ is the local Maxwellian distribution function given by
\begin{equation}
    \label{eq:maxwellian}
	f^{(0)}(\vec{x},\vec{v},t) \coloneqq \frac{\rho(\vec{x},t)}{m}\frac{1}{(\frac{4}{3}\pi e(\vec{x},t))^{3/2}} \exp\left(-\frac{\abs{\vec{v} - \vec{u}(\vec{x},t)}^2}{\frac{4}{3}e(\vec{x},t)}\right),
\end{equation}
where $\rho(\vec{x},t)$ is the local density, $\vec{u}(\vec{x},t)$ is the local mean velocity, and $e(x,t)$ is the internal energy per unit mass, see \eqref{eq:matching_moments}--\eqref{def:energy}.
	\begin{remark}
        The local Maxwellian is the unique function satisfying $Q(f^{(0)},f^{(0)})=0$ and the constraints 
		\begin{equation}
			\label{eq:matching_moments}
			\int_{{\mathbb{R}}^3}
			m\begin{pmatrix}
			1\\
			\vec v\\ 
			\abs{\vec v}^2
			\end{pmatrix} f^{(0)}\,d\vec v 
            = 
            \begin{pmatrix}
                \rho\\
                \rho \vec u\\
                \rho \frac{1}{2}\abs{\vec u}^2 + \rho e
            \end{pmatrix}
            \coloneqq \int_{{\mathbb{R}}^3}
			m\begin{pmatrix}
			1
			\\
			\vec v\\ 
			\abs{\vec v}^2
			\end{pmatrix} f\,d\vec v.
		\end{equation}
        The result that the Maxwellian is uniquely determined by the constraints \eqref{eq:matching_moments} comes as a consequence of the fact that the Boltzmann collision operator has only five independent collision invariants, which is far from trivial.
        Its modern proofs are based on a Gr\"onwall approach \cite{carleman, grad,collisionInvariants, arkeryd, arkerydCercignani}. 
	\end{remark}

The BGK collision operator \eqref{eq:bgk} preserves two important properties of the true collision operator \eqref{eq:full_collision}: it relaxes to the same (equilibrium) Maxwellian distribution, and it possesses the same collision invariants. It also greatly simplifies calculations.
However, there are significant differences: for example the BGK operator predicts an incorrect Prandtl number for monatomic gases \cite{Harris}.
We expect that all of the same calculations in this manuscript could in principle be carried out for the full collision operator, giving qualitatively similar but quantitatively different results.

Assuming all microscopic constituents to be identical with mass $m$, the local macroscopic quantities are defined as moments of the distribution function, i.e.,
\begin{equation}
	\rho(\vec{x},t) = \int_{\mathbb{R}^3} m f(\vec{x},\vec{v},t) d\vec{v},\qquad \rho(\vec{x},t) \vec{u}(\vec{x},t) = \int_{\mathbb{R}^3} m \vec{v} f(\vec{x},\vec{v},t) d\vec{v},\label{eq:newj-1}
\end{equation}
and the specific internal energy and the specific entropy are defined as
\begin{align}
\label{def:energy}
	\rho(\vec{x},t) e(\vec{x},t) &= \int_{\mathbb{R}^3} \frac{m}{2} \abs{\vec{v} - \vec{u}(\vec{x},t)}^2 f(\vec{x},\vec{v},t) d\vec{v},\\
	\label{def:entropy}
	\rho(\vec{x},t) \eta(\vec{x},t) &= - \int_{\mathbb{R}^3} k_B \log \left(f(\vec{x},\vec{v},t)\right) f(\vec{x},\vec{v},t) d\vec{v},
\end{align}
where $k_B$ is the Boltzmann constant.

Another important and measurable macroscopic property is the temperature $\theta=\theta(\vec{x}, t)$, measured in Kelvin. This will be defined below from the equation of state. For the monatomic case under consideration it is related to the internal energy by the equipartition of energy theorem \cite[Section 15.8]{fasanoMarmi}
    \begin{equation} \label{eq:equipartition_of_energy}
        e = \frac{3}{2}{R_s} \theta \rev{=\frac{3}{2}\frac{k_B}{m} \theta},
    \end{equation}
    where ${R_s}\coloneqq \frac{k_B}{m}$ is the so-called specific gas constant.
    The physical intuition behind the equipartition of energy theorem is that the energy is equally distributed across all the degrees of freedom available to the constituents of the system; in the case of spherical monatomic gas molecules, these are the three directions of motion of the center of mass.

It is important to choose the relaxation time $\tau$ of the BGK operator to match that of the Boltzmann operator it aims to approximate.
In particular, for the Boltzmann collision operator for hard spheres \cite[Section 3.3]{Harris}, i.e.~$\mathcal{B}(\beta, \varphi, \abs{\vec{v}-\vec{v}_*})\coloneqq 4r_0\abs{\vec{v}-\vec{v}_*}\sin(\theta)\cos(\theta)$, we can evaluate the collision frequency near the Maxwellian to be
\begin{equation}
    \label{eq:scalingT}
    \frac{1}{\tau} = \int_{\mathbb{R}^3} \int_{0}^{r_0}\!\!\int_{0}^{2\pi}\!\!\!\int_{0}^{\pi} \mathcal{B}(\beta, \varphi, \abs{\vec{v}-\vec{u}}) f^{(0)}\,d\beta\,d\varphi\, d\vec{v} \propto \int_{\mathbb{R}^3} \abs{\vec{v}-\vec{u}}f^{(0)}\,d\vec{v} \propto \rho \sqrt{e},
\end{equation}
where the proportionality constants only depend on the diameter of the hard spheres $r_0$ and the number of degrees of freedom associated with each particle.
\begin{remark}
    We expect the same theory developed in this paper for hard molecules to hold for more general \rev{hard} collision kernels.
    In particular the calculations here presented should carry over to a power-law collision kernel of the form $\mathcal{B}(\beta, \varphi, \abs{\vec{v}-\vec{v}_*})\coloneqq b(\theta, \varphi)\abs{\vec{v}-\vec{v}_*}^\lambda$ for some $\lambda \geq 1$. \rev{See the discussion in \ref{sec:full_collision} for more details.}
\end{remark}

\section{Deriving balance laws from kinetic theory}
We recall the standard derivation of the balance laws of continuum mechanics from kinetic theory \cite[Chapter 5]{cercignani}.
It is well-known that the Boltzmann collision operator \eqref{eq:full_collision} conserves mass, momentum, and energy \cite{cercignani,Harris,collisionInvariants}. 
The BGK operator \eqref{eq:bgk} possesses the same collision invariants as the full one.
That is, for $\psi\in \left\{m, m\vec{v}, \frac{1}{2}m\abs{\vec{v}}^2\right\}$:
\begin{equation}
    \int_{\mathbb{R}^3} \psi(\vec{v}) Q(f,f) d\vec{v} = \int_{\mathbb{R}^3} \psi(\vec{v}) Q_{\text{BGK}}(f, f^{(0)}) d\vec{v} = 0.
\end{equation}
Multiplying the Boltzmann equation \eqref{eq:boltzmann} by the collision invariant $\psi(\vec{v})\equiv m$ and integrating the result over the velocity space{,} we obtain the continuity equation, also known as the mass balance equation:
\begin{equation} \label{eq:balance_law_mass}
    \frac{\partial}{\partial t} \int_{\mathbb{R}^3} mf d\vec{v} + \nabla_x \cdot \int_{\mathbb{R}^3} m\vec{v} f d\vec{v} = 0 \Rightarrow \frac{\partial \rho}{\partial t} + \nabla_x \cdot (\rho \vec{u}) = 0.
\end{equation}

Similarly, with $\psi(\vec{v}) \equiv m\vec{v}$, we obtain, after a little algebra, the linear momentum balance equation. We briefly sketch the derivation, since we will employ similar steps below.
Multiplying the Boltzmann equation \eqref{eq:boltzmann} by $m\vec{v}$ and integrating over the velocity space yields 
\begin{equation}
    \frac{\partial}{\partial t} \int_{\mathbb{R}^3} m\vec{v} f d\vec{v} + \nabla_x \cdot \int_{\mathbb{R}^3} m\vec{v}\otimes \vec{v} f d\vec{v} = 0.\label{eq:newj0}
\end{equation}
Introducing the vector field $\vec{w} = \vec{w}(\vec{x},\vec{v},t)$ as 
\begin{equation}
    \vec{w} \coloneqq \vec{v} - \vec{u} \qquad \textrm{i.e. } \vec{w} (\vec{x},\vec{v},t) \coloneqq \vec{v} - \vec{u}(\vec{x},t), \label{eq:newj1}
\end{equation}
and replacing $\vec{v}$ by $\vec{w} + \vec{u}$ in the second term of \eqref{eq:newj0}, using also the fact that $\int_{\mathbb{R}^3}   m(\vec{v}-\vec{u}) f d\vec{v}=0$, we rewrite \eqref{eq:newj0} as 
\begin{equation}
	\frac{\partial}{\partial t} (\rho \vec{u}) + \nabla_x \cdot \left( \int_{\mathbb{R}^3} m \vec{w} \otimes \vec{w} f d\vec{v} \right) + \nabla_{\vec{x}}\cdot \left(\rho \vec{u}\otimes \vec{u}\right) = 0.\label{eq:newj2}
\end{equation}
Introducing then the Cauchy stress tensor $\mathbb{T}$ through 
\begin{equation}
    \mathbb{T}(\vec{x},t) \coloneqq - \int_{\mathbb{R}^3} m\vec{w}(\vec{x},\vec{v},t)\otimes\vec{w}(\vec{x},\vec{v},t)f(\vec{x},\vec{v},t) d\vec{v}, 
    \label{eq:newj3}
\end{equation}
we finally obtain the linear momentum balance in the form
\begin{equation}
	\frac{\partial}{\partial t} (\rho \vec{u}) + \nabla_{\vec{x}}\cdot \left(\rho \vec{u}\otimes \vec{u}\right) - \nabla_{\vec{x}} \cdot \mathbb{T} = 0,
\end{equation}
or alternatively, with help of mass balance equation \eqref{eq:balance_law_mass},

\begin{equation} \label{eq:balance_momentum}
    \rho \left( \frac{\partial \vec{u}}{\partial t} + (\nabla_{\vec{x}} \vec{u}) \vec{u} \right) - \nabla_{\vec{x}} \cdot \mathbb{T} = 0.
\end{equation}

We are left to derive the balance equation of energy by multiplying the Boltzmann equation \eqref{eq:boltzmann} by $\psi(\vec{v}) \equiv \frac{1}{2}m\abs{\vec{v}}^2$ and integrating over the velocity space.
Proceeding as before we start from
\begin{equation}
	\partial_t \mathcal{E} + \nabla_x \cdot \mathcal{F} = 0,\label{eq:newj8}
\end{equation}
where the total energy $\mathcal{E}=\mathcal{E}(\vec{x},t)$ and the energy flux $\mathcal{F}=\mathcal{F}(\vec{x},t)$ are defined as
\begin{align}
    \mathcal{E}(\vec{x},t) &\coloneqq \int_{\mathbb{R}^3} \frac{1}{2} m \abs{\vec{v}}^2 fd\vec{v}, 
    \qquad \mathcal{F}(\vec{x},t) \coloneqq \int_{\mathbb{R}^3} \frac{1}{2} m \abs{\vec{v}}^2 \vec{v} f d\vec{v}.
\end{align}
Replacing $\vec{v}$ by $\vec{w} + \vec{u}$, see \eqref{eq:newj1}, and using \eqref{eq:newj-1}--\eqref{def:entropy} and  
$\int_{\mathbb{R}^3}   m \vec{w} f d\vec{v}=0$, we obtain
\begin{equation}
    \begin{split}
	\mathcal{E} &= \int_{\mathbb{R}^3} m \, \frac{\abs{\vec{w+u}}^2}{2} f d\vec{v} \\& = \int_{\mathbb{R}^3} m \, \frac{\abs{\vec{w}}^2}{2} f d\vec{v} + \vec u\cdot\int_{\mathbb{R}^3} m \vec{w} f d\vec{v} + \frac{\abs{\vec{u}}^2}{2}\int_{\mathbb{R}^3} m f d\vec{v} \,  \\
	&= \rho e + \rho \frac{\abs{\vec{u}}^2}{2}.
    \end{split} \label{eq:newj5}
\end{equation}
Proceeding similarly, we rewrite the energy flux $\mathcal{F}$ as
\begin{align}\label{eq:newj6}
\begin{split}
     \mathcal{F}&= \int_{\mathbb{R}^3} m\, \frac{ \abs{\vec{w}+\vec{u}}^2}{2} (\vec{w}+\vec{u}) f d\vec{v} \\
     &= \int_{\mathbb{R}^3} m \left(\frac{\abs{\vec{w}}^2}{2} + \vec{w}\cdot \vec{u} + \frac{\abs{\vec{u}}^2}{2}\right) (\vec{w}+\vec{u}) f d\vec{v} \\
			   &=  \int_{\mathbb{R}^3} m \, \vec{w}\frac{\abs{\vec{w}}^2}{2} f d\vec{v} + \left(\int_{\mathbb{R}^3} m \,(\vec{w}\otimes\vec{w}) f d\vec{v} \right)\vec{u} + \frac{\abs{\vec{u}}^2}{2}\int_{\mathbb{R}^3} m \,\vec{w} f d\vec{v}  \\ 
              &\quad + \vec{u}\int_{\mathbb{R}^3} m \frac{\abs{\vec{w}}^2}{2} f d\vec{v} \,   + (\vec{u}\otimes\vec{u})\int_{\mathbb{R}^3} m \vec{w} f d\vec{v} + \frac{\abs{\vec{u}}^2}{2}\, \vec{u}\int_{\mathbb{R}^3} m f d\vec{v}\, \\
			   &= \vec{Q} \uze{\,-\,} \mathbb{T}\vec{u} + \rho e \vec{u} + \rho \frac{\abs{\vec{u}}^2}{2} \vec{u}, 
\end{split}
\end{align}
where we employed again the identity $\int_{\mathbb{R}^3}m\vec{w}f d\vec v = 0$, and where we defined the heat flux vector $\vec{Q}$ as
\begin{equation}\label{p:energyflux}
    \vec{Q}(\vec{x},t) \coloneqq \int_{\mathbb{R}^3} \frac{1}{2} m \abs{\vec{w}(\vec{x},\vec{v}, t)}^2 \vec{w}(\vec{x},\vec{v}, t) f(\vec{x},\vec{v},t) d\vec{v}.
\end{equation}
Thus we can rewrite \eqref{eq:newj8} as 
\begin{equation}
	\frac{\partial}{\partial t} \left( \frac{1}{2} \rho \abs{\vec{u}}^2 + \rho e \right) + \nabla_{\vec{x}} \cdot \left( \frac{1}{2} \rho \abs{\vec{u}}^2 \vec{u} + \rho e \vec{u} {\,-\,} \mathbb{T} \vec{u} + \vec{Q} \right) = 0.
\end{equation}
Using the balance equations for mass and linear momentum, we can rearrange the terms to obtain a more familiar form of the balance equation for energy as the evolution equation for internal energy:
\begin{equation} \label{eq:evolution_internal_energy}
    \rho \left( \frac{\partial e}{\partial t} + \vec{u}\cdot \nabla_{\vec{x}} e \right) {\,-\,} \mathbb{T} : \nabla_{\vec{x}} \vec{u} + \nabla_{\vec{x}} \cdot \vec{Q} = 0.
\end{equation}

Finally, we would like to derive the balance equation of entropy. To this end, we multiply the Boltzmann equation \eqref{eq:boltzmann} by $\psi(\vec{x}, \vec{v}, t) = {-}k_B\rev{(\log f(\vec{x},\vec{v},t)+1)}$ and integrate over the velocity space to obtain
\begin{equation}
    \frac{\partial}{\partial t} \int_{\mathbb{R}^3} {-} k_B f \log f d\vec{v} + \nabla_{\vec{x}} \cdot \int_{\mathbb{R}^3} {-}k_B \vec{v} f \log f d\vec{v} = {-}\int_{\mathbb{R}^3} k_B \, Q_{\text{BGK}}(f, f^{(0)}) \log f d\vec{v}, \label{eq:newj9}
\end{equation}
where the right-hand side does not vanish in this case since ${-}\log f(\vec{x},\vec{v},t)$ is not a collision invariant. Using again the decomposition $\vec{v}=\vec{u}+\vec{w}$ and the definition of entropy density \eqref{def:entropy}, we can rewrite 
\eqref{eq:newj9} in the form of balance equation for entropy:
\begin{equation}
    \label{eq:almostClausiusDuhem}
    \rho \left( \frac{\partial \eta}{\partial t} + \vec{u}\cdot \nabla_{\vec{x}} \eta \right) + \nabla_{\vec{x}} \cdot \vec{\Phi} = \xi\ ,
\end{equation}
with the entropy flux $\vec{\Phi}$ defined as
\begin{equation}
\vec{\Phi}(\vec{x},t) \coloneqq \int_{\mathbb{R}^3} {-} k_B \, \vec{w}(\vec{x},\vec{v},t) f(\vec{x},\vec{v},t) \log f(\vec{x},\vec{v},t) d\vec{v},
\end{equation}
and the entropy production $\xi$ defined as
\begin{align}
\label{eq:BGK_entropy_production}
\xi(\vec{x},t) &\coloneqq \int_{\mathbb{R}^3} {-}k_B \, Q_{\text{BGK}}(f(\vec{x},\vec{v},t),f^{(0)}(\vec{x},\vec{v},t)) \log f(\vec{x},\vec{v},t) d\vec{v}\\
&= \frac{1}{\tau} \int_{\mathbb{R}^3} \, k_B {(f - f^{(0)})} \log f d\vec{v}.
\end{align}

\begin{remark}
The balance equation for entropy \eqref{eq:almostClausiusDuhem} is reminiscent of the Clausius--Duhem inequality from continuum thermodynamics \cite[Chapter I, Eq.~I.34]{Truesdell}. In fact, we will show in the next section {that} the entropy production rate $\xi$ is non-negative and later we will show when we can recover the classical form of the Clausius--Duhem inequality, determining the entropy flux vector $\vec{\Phi}(\vec{x},t)=\vec{Q}(\vec{x},t)/\theta(\vec{x},t)$.
\end{remark}

\section{H-Theorem and entropy production}
A key aspect of the Boltzmann equation is the celebrated H-theorem, which states that the entropy of a closed system never decreases over time, i.e.~the Boltzmann equation is irreversible. This result is classical, but we include it here for completeness of exposition.
Starting from the definition of the entropy production rate $\xi(\vec{x},t)$ given in \eqref{eq:BGK_entropy_production} and adding and subtracting the term {$\int_{\mathbb{R}^3} k_B(f^{(0)}{-}f) \log f^{(0)} d\vec{v}$} we obtain
\begin{equation} \label{eq:expression_for_xi}
    \xi(\vec{x},t) \!=\!\frac{1}{\tau}\! \left( \int_{\mathbb{R}^3}\!\! k_B (f - f^{(0)}) \left(\log f - \log f^{(0)} \right) d\vec{v} - \int_{\mathbb{R}^3}\!\! k_B (f^{(0)} - f) \log f^{(0)} d\vec{v} \right).
\end{equation}
We now show that the second integral vanishes because the BGK collision operator conserves mass, momentum, and energy. Using the explicit expression \eqref{eq:maxwellian} for the Maxwellian $f^{(0)}$ we can compute
\begin{equation}\label{TTT}
	\log f^{(0)}  = \gamma - \frac{\abs{\vec{w}}^2}{\frac{4}{3}e} \quad \rev{\textrm{ with } \gamma :=\log\left(\frac{\rho}{m}\right)  - \frac{3}{2} \log\left(\frac{4\pi e}{3}\right)},
\end{equation}
where $\gamma= \gamma(\vec x,t)$ depends on space and time but not velocity. \rev{Since, by \eqref{eq:matching_moments}, 
$f$ and $f^{(0)}$ have the same first three moments with respect to velocity, we observe that} 
\begin{align}
	\int_{\mathbb{R}^3}\!\! (f^{(0)} \!-\! f)\! \log f^{(0)} d\vec{v} \!=\! \gamma \!\!\int_{\mathbb{R}^3}\!\!(f^{(0)}\! -\! f) d\vec{v} - \frac{3}{4e}\!\int_{\mathbb{R}^3}\!\! \abs{\vec{w}}^2 (f^{(0)}\! -\! f) d\vec{v} \!=\! 0.
\end{align}

We next consider the first term in \eqref{eq:expression_for_xi}. Using the monotonicity property of the logarithm, i.e.~$(f - f^{(0)})\big(\log(f) - \log(f^{(0)})\big) \geq 0$, we conclude that the entropy production rate is always non-{negative}. Thus we have
\begin{equation}
    \label{eq:entropyBalanceLaw}
    \rho \left( \frac{\partial \eta}{\partial t} + \vec{u}\cdot \nabla_{\vec{x}} \eta\right) + \nabla_{\vec{x}} \cdot \vec{\Phi} = \xi\quad \text{where}\quad \xi \geq 0.
\end{equation}
We have thus shown that the entropy production rate is always non-negative, which is a particular form of the second law of thermodynamics.

\begin{remark}
The analogous result can be obtained for the full Boltzmann collision operator by using well-known properties of the scattering kernel $B(\theta, \varphi, \abs{\vec{v}-\vec{v}_*})$ \cite[Section 1.7]{cercignani}.
\end{remark}
\begin{remark}[Various formulations of the second law of thermodynamics]
    \label{rmk:second_law_formulations}
    Considering $\Omega$ an open bounded subdomain of $\mathbb{R}^3$ constant in time, we introduce the integrated entropy $H$ as
    \begin{equation}
        H(t) \coloneqq \int_{\Omega} \rho(\vec{x},t) \eta(\vec{x},t) d\vec{x},
    \end{equation}
    Integrating the entropy balance law \eqref{eq:entropyBalanceLaw} over $\Omega$ we obtain
    \begin{equation}
        \label{eq:thermal_bath}
        \frac{d}{dt} H(t) \geq \int_{\partial \Omega} (\vec{\Phi} \cdot \vec{n}) \,dS.
    \end{equation}
    Considering an isolated system, defined via the condition $\vec{\Phi} \cdot \vec{n} = 0$ on $\partial \Omega$, we obtain 
    \begin{equation}
        \label{eq:clausius_inequality}
        \frac{d}{dt} H(t) \geq 0,
    \end{equation}
    yet another form of the second law of thermodynamics.
    Thermodynamic isolation can be achieved by many different boundary conditions \cite[Section 1.7]{cercignani}.

    The second law of thermodynamics has been the subject of much debate since its inception in the $19^{\text{th}}$ century.
    Many authors have argued in favor of different formalizations of the intuitive notion behind the second law of thermodynamics.
    We have presented here several different formalizations of the second law of thermodynamics: 
    one is the fact that the entropy of an isolated system never decreases over time (see \eqref{eq:clausius_inequality}) and is the famous H-theorem.
    In Truesdell's terminology \cite[Chapter I, Section 5]{Truesdell}, this is known as the Clausius' inequality.
    Another form we have presented is the fact that the entropy of a system is governed by \eqref{eq:thermal_bath}, which is a generalization of what Truesdell refers to as the heat-bath inequality \cite[Chapter I, Section 5]{Truesdell}.
    Rational thermodynamics features a third formalization of the second law of thermodynamics, known as the Clausius--Duhem inequality, which we will recover in \Cref{sec:EoS} below.
    Using the Clausius--Duhem inequality, the heat-bath inequality presented in \cite[Chapter 1, Section 5]{Truesdell} can be retrieved from \eqref{eq:thermal_bath}.

    The Rajagopal--Srinivasa approach starts with the general formulation of the balance equation of entropy \eqref{eq:entropyBalanceLaw} and then postulates non-negative constitutive equations for the entropy production $\xi$, which can be seen as yet another formalization of the second law of thermodynamics.
    This formalization of the second law of thermodynamics is very general and implies all of the other formalizations we have presented so far, as well as the Boltzmann inequality, i.e.~the fact that the local entropy production is non-negative.

\end{remark}

\section{Chapman--Enskog approach to constitutive equations for entropy and entropy production}
\label{sec:EoS}
We will assume that near thermodynamic equilibrium the distribution function $f$ can be approximated as a perturbation of the local Maxwellian equilibrium distribution function $f^{(0)}$.
This intuition is formalized by the Chapman--Enskog expansion \cite{cercignani,Chapman}, i.e.,
\begin{equation}
    \label{eq:CEexpansion}
    f(\vec{x},\vec{v},t) = f^{(0)}(\vec{x},\vec{v},t) + (\mathrm{Kn}) f^{(1)}(\vec{x},\vec{v},t) + (\mathrm{Kn})^2 f^{(2)}(\vec{x},\vec{v},t) + \cdots,
\end{equation}
where $\mathrm{Kn} = \frac{\lambda}{L}$ is the Knudsen number, the ratio between the mean free path of the particles $\lambda$ and a characteristic macroscopic length scale $L$ of the system under consideration.
The Chapman--Enskog expansion is a formal power series of the distribution function $f$ around the local Maxwellian $f^{(0)}$ in terms of the Knudsen number $\mathrm{Kn}$.

Using the relation between the mean free path, the thermal velocity, most probable speed $V_{\theta} \coloneqq \rev{(2R_s\theta)^{\frac{1}{2}}}$, and the average relaxation time $\bar{\tau}$,
\begin{equation}
    \label{eq:tau_Kn}
    \lambda = V_{\theta} \bar{\tau} \quad\Rightarrow\quad \mathrm{Kn} = \frac{\bar{\tau} V_{\theta}}{L}, \quad \bar{\tau} = \frac{\mathrm{Kn} \, L}{V_{\theta}},
\end{equation}
we can express the relaxation time in terms of $\mathrm{Kn}$.
\subsection{Chapman--Enskog expansions and its consequences for macroscopic moments}
In order to use the approximation of $f$ presented in \eqref{eq:CEexpansion} we need to determine $f^{(k)}$ for $k\geq 1$. This is done recursively starting from $k=0$.
To determine $f^{(1)}$ we introduce two time variables, corresponding to different time scales, i.e.~$t_0=t$ and $t_1=\mathrm{Kn}\, t \ll t_0$, and we will treat $f$ as a function of the two time scales independently \cite[Chapter 4]{dellar}, i.e.~$f=f(\vec{x},\vec{v},t_0,t_1)$.
It is common to refer to $t_0$ as the advective time scale, while we will call $t_1$ the diffusive or viscous time scale. Splitting the time derivative in terms of $t_0$ and $t_1$,i.e.~replacing $\partial_t\rightarrow \partial_{t_0} + \mathrm{Kn}\,\partial_{t_1}$, we rewrite the Boltzmann equation \eqref{eq:boltzmann} with \eqref{eq:bgk} as
\begin{equation}
    \label{eq:BoltzmannMultipleScales}
    \partial_{t_0}f + \mathrm{Kn}\,\partial_{t_1}f + \vec{v} \cdot \nabla_{\vec{x}} f = - \frac{1}{\mathrm{Kn}\,\tilde{\tau}}(f-f^{(0)})
\end{equation}
where $\tilde{\tau}$ is the relaxation time rescaled by the Knudsen number, i.e.~$\tau=\mathrm{Kn}\,\tilde{\tau}$.

Expanding $f$ with the Chapman--Enskog ansatz \eqref{eq:CEexpansion}
does not yield a unique set of evolution equations for $f^{(0)}$ and $f^{(1)}$.
To retrieve uniqueness after the expansion, it is standard to demand that hydrodynamical quantities evolve slowly while microscopic quantities relax quickly. This translates to the assumption that
\begin{equation}
    \label{eq:solvability}
    \textrm{for all } k\in\mathbb{N}_{\ge 1}: \quad \int_{\mathbb{R}^3} f^{(k)}\, d\vec{v} = 0, \quad \int_{\mathbb{R}^3} \vec{v}f^{(k)}\,d\vec{v}=0, \quad \int_{\mathbb{R}^3}\abs{\vec{v}}^2f^{(k)}\,d\vec{v} = 0. 
\end{equation}
In other words, the higher-term corrections relax so quickly so as not to contribute to the macroscopic observables. Assumption \eqref{eq:solvability} is referred to as the {\it solvability condition} \cite[Section 5.3]{cercignani2}.
\begin{remark}
    Note that the solvability condition is very natural in view of the property of the Maxwellian $f^{(0)}$ capturing the first three moments, with respect to the velocity, of the full distribution function $f$, as in \eqref{eq:matching_moments}.
    The solvability condition then extends this assertion to all truncations of the Chapman--Enskog expansion of $f$.
\end{remark}
\begin{remark}
    Note that the solvability condition \eqref{eq:solvability} implies that both the first moment of $f^{(k)}$ with respect to $\vec{w}$ vanishes, i.e.~for all $k\in \mathbb{N}_{\ge 1}$: 
    \begin{equation}
        \int_{\mathbb{R}^3} \vec{w} f^{(k)} d\vec{v} = \int_{\mathbb{R}^3} (\vec{v} - \vec{u}) f^{(k)} d\vec{v} = \int_{\mathbb{R}^3} \vec{v} f^{(k)} d\vec{v} - \vec{u} \int_{\mathbb{R}^3} f^{(k)} d\vec{v} = 0,
    \end{equation}
    and the second moment of $f^{(k)}$ with respect to $\vec{w}$ vanishes, i.e.~for all $k\in \mathbb{N}_{\ge 1}$: 
    \begin{equation}
        \int_{\mathbb{R}^3} \abs{\vec{w}}^2 f^{(k)} d\vec{v} = \int_{\mathbb{R}^3} \left( \abs{\vec{v}}^2 - 2(\vec{v}\cdot \vec{u}) + \abs{\vec{u}}^2 \right) f^{(k)} d\vec{v} = 0.
    \end{equation}
\end{remark}

For the moment being what we care about is that \eqref{eq:solvability} carries an important implication for the macroscopic quantities $\rho(\vec{x},t)$, $\vec{u}(\vec{x},t)$, $e(\vec{x},t)$, and $\eta(\vec{x},t)$, and also the {heat and} entropy flux{es $\vec{Q}$ and} $\vec{\Phi}$.

\begin{definition}[Expansion independence]
Let $\varphi[f]$ be a macroscopic observable computed from $f$. We say that $\varphi$ is expansion-independent if
\begin{equation}
\textrm{for all } k\in \mathbb{N}: \quad \varphi[f] = \varphi\left[\sum_{i=0}^k (\mathrm{Kn})^i f^{(i)}\right],
\end{equation}
i.e.~the quantity is unchanged by evaluating it on any partial sum of the Chapman--Enskog series.
\end{definition}

Henceforth we denote by $\cdot^{(k)}$ a macroscopic quantity computed using the Chapman--Enskog series up to order $k$, i.e.~$\sum_{i=0}^{k} (\mathrm{Kn})^{i} f^{(i)}$.

\begin{proposition}
Each of $\rho(\vec{x},t)$, $\vec{u}(\vec{x},t)$, $e(\vec{x},t)$ are expansion-independent under the solvability condition \eqref{eq:solvability}.

\end{proposition}
\begin{proof}
    From the solvability condition \eqref{eq:solvability} we trivially see that the density is expansion-independent, i.e.
    \begin{align}
        \rho{^{(k)}}(\vec{x},t) &= \int_{\mathbb{R}^3} {m}\sum_{i=0}^k (\mathrm{Kn})^i f^{(i)}(\vec{x},\vec{v},t) d\vec{v} \\
        &= \int_{\mathbb{R}^3} {m} f^{(0)}(\vec{x},\vec{v},t) d\vec{v} + \int_{\mathbb{R}^3}{m} \sum_{i=1}^k (\mathrm{Kn})^i f^{(i)}(\vec{x},\vec{v},t) d\vec{v}\\
        &=\int_{\mathbb{R}^3} {m} f^{(0)}(\vec{x},\vec{v},t) d\vec{v} {= \rho^{(0)}(\vec x, t) = \rho(\vec x, t)},
    \end{align}
    {where the last equality follows from \eqref{eq:matching_moments}.}    
    
    Similarly, using also \eqref{eq:matching_moments}, for the macroscopic velocity we have
    \begin{align}
        (\rho(\vec{x},t) \vec{u}(\vec{x},t)){^{(k)}} &= \int_{\mathbb{R}^3} {m}\vec{v} \sum_{i=0}^k (\mathrm{Kn})^i f^{(i)}(\vec{x},\vec{v},t) d\vec{v} \\
        &= \int_{\mathbb{R}^3} {m}\vec{v} f^{(0)}(\vec{x},\vec{v},t) d\vec{v} + \int_{\mathbb{R}^3} {m}\vec{v} \sum_{i=1}^k (\mathrm{Kn})^i f^{(i)}(\vec{x},\vec{v},t) d\vec{v}\\
        &= \int_{\mathbb{R}^3} {m}\vec{v} f^{(0)}(\vec{x},\vec{v},t) d\vec{v} { = (\rho(\vec x, t)\vec u(\vec x,t))^{(0)} = \rho(\vec x, t)\vec u(\vec x,t)}.
    \end{align}
    
    Lastly, for the {specific} internal energy 
     we have
    \begin{align}
        (\rho(\vec{x},t) e(\vec{x},t)){^{(k)}} &= \int_{\mathbb{R}^3} \frac{1}{2} m \abs{\vec{w}}^2 \sum_{i=0}^k (\mathrm{Kn})^i f^{(i)}(\vec{x},\vec{v},t) d\vec{v} \\
        &= \int_{\mathbb{R}^3}\frac{1}{2}m\left[ \abs{\vec{w}}^2 f^{(0)}(\vec{x},\vec{v},t) + \abs{\vec{w}}^2 \sum_{i=1}^k (\mathrm{Kn})^i f^{(i)}(\vec{x},\vec{v},t)\right] d\vec{v}\\
        &= \int_{\mathbb{R}^3} \frac{1}{2} m \abs{\vec{w}}^2 f^{(0)}(\vec{x},\vec{v},t) d\vec{v} {= (\rho(\vec x, t)e(\vec x,t))^{(0)}} \\
        &{= \rho(\vec x, t)e(\vec x,t)}        
       \end{align}
{once more based on \eqref{eq:matching_moments}.}

\end{proof}

\begin{proposition}
\label{proposition-entropy}
Under the solvability condition \eqref{eq:solvability}, the specific entropy is such that $\rev{\rho}\eta^{(1)} = \rev{\rho} \eta^{(0)}+\mathcal{O}(\mathrm{Kn}^2)$.
\end{proposition}

\begin{proof}
\rev{Starting from the definition of $\rho\eta^{(1)}$, i.e.,}  
\begin{equation}
    \rho \eta^{(1)} =\int_{\mathbb{R}^3} {-}k_B \,\left( f^{(0)} + \mathrm{Kn} f^{(1)} \right) \log\left( f^{(0)} + \mathrm{Kn} f^{(1)} \right) d\vec{v},
\end{equation}
and using the Taylor expansion of the logarithm
\begin{equation}\label{TTT1}
    \log(f^{(0)} + \mathrm{Kn} f^{(1)}) = \log(f^{(0)}) + \mathrm{Kn}\, \frac{f^{(1)}}{f^{(0)}} + \mathcal{O}(\mathrm{Kn}^2)
\end{equation}
we can rewrite the {specific} entropy $\rev{\rho\eta^{(1)}}$ as
\begin{align}
    \rho \eta^{(1)} &= -\int_{\mathbb{R}^3} {k_B}f^{(0)} \log f^{(0)} d\vec{v} - \mathrm{Kn} \int_{\mathbb{R}^3} k_B f^{(1)} \log f^{(0)} d\vec{v} + \mathcal{O}(\mathrm{Kn}^2)\\
    &= \rho \eta^{(0)} {-} \, \gamma \mathrm{Kn} \int_{\mathbb{R}^3} \rev{k_B} f^{(1)} d\vec{v} {+} \frac{{3}\mathrm{Kn}}{{4}e}\!\int_{\mathbb{R}^3}\!\! k_B \, \abs{\vec{w}}^2 f^{(1)} d\vec{v}
    + \mathcal{O}(\mathrm{Kn}^2)\\
    &{= \rho \eta^{(0)} + \mathcal{O}(\mathrm{Kn}^2)},   
\end{align}
where \rev{we used \eqref{TTT} and the fact that} the last two integrals vanish due to the solvability condition \eqref{eq:solvability}.
\end{proof}

\begin{proposition}
\label{prop:evenness_entorpy_flux}
Under the solvability condition \eqref{eq:solvability}, the entropy  \rev{ and energy fluxes $\vec{\Phi}^{(k)}$ and $\vec{Q}^{k}$, $k=0,1$, satisfy} 
\begin{equation}\label{TTT2}
\rev{\vec{\Phi}^{(0)} = \frac32 \frac{R_s}{e}\vec{Q}^{(0)} \quad \textrm{ and } \quad \vec{\Phi}^{(1)}=\frac32 \frac{R_s}{e}\vec{Q}^{(1)} + \mathcal{O}(\mathrm{Kn}^2).}
\end{equation}
\end{proposition}
\begin{proof}
We \rev{start from the definitions of $\vec{\Phi}^{(k)}$, $k=0,1$, and use \eqref{TTT} and \eqref{eq:solvability}. This yields}
\begin{subequations}
\begin{align}
    \rev{\vec{\Phi}^{(0)}} &= \rev{{-}{k_B} \int_{\mathbb{R}^3} \vec{w} f^{(0)} \log\left( f^{(0)}\right) d\vec{v}} \\ &= \rev{{-}{k_B} \gamma \int_{\mathbb{R}^3} \vec{w} f^{(0)} d\vec{v}  + \frac{3k_B}{2e} \int_{\mathbb{R}^3} \vec{w} \frac{|\vec{w}|^2}{2} f^{(0)} d\vec{v}} \\
    &= \rev{\frac{3k_B}{2me} \int_{\mathbb{R}^3} m\,\vec{w} \frac{|\vec{w}|^2}{2} f^{(0)} d\vec{v} = \frac{3R_s}{2e} \vec{Q}^{(0)}},\label{TTT3}
\end{align}
\rev{and, using also \eqref{TTT1},}
\begin{align}
    \rev{\vec{\Phi}^{(1)}} &= \rev{\int_{\mathbb{R}^3} {-}{k_B}\vec{w} \left( f^{(0)} + \mathrm{Kn}\, f^{(1)} \right) \log\left( f^{(0)} + \mathrm{Kn}\, f^{(1)} \right) d\vec{v}}, \\
    &= \int_{\mathbb{R}^3} {-}{k_B}\vec{w} \left( f^{(0)} + \mathrm{Kn}\, f^{(1)} \right) \left[ \log(f^{(0)}) + \mathrm{Kn} \frac{f^{(1)}}{f^{(0)}} +\mathcal{O}(\mathrm{Kn}^2)\right]\,d\vec{v}\\
    &= \vec{\Phi}^{(0)} {-} \mathrm{Kn} \int_{\mathbb{R}^3} {k_B}\vec{w} f^{(1)} \log(f^{(0)}) \,d\vec{v} {-} \mathrm{Kn}\int_{\mathbb{R}^3} {k_B}\vec{w} f^{(1)} \,d\vec{v} + \mathcal{O}(\mathrm{Kn}^2) \\
    &= \rev{\vec{\Phi}^{(0)} {-} \mathrm{Kn} \int_{\mathbb{R}^3} {k_B}\vec{w} f^{(1)} \left[\gamma - \frac{|\vec{w}|^2}{\frac43 e}\right] \,d\vec{v} + \mathcal{O}(\mathrm{Kn}^2)} \\
    &= \rev{\vec{\Phi}^{(0)} + \mathrm{Kn} \frac{3k_B}{2me} \int_{\mathbb{R}^3} m\, \vec{w} \frac{|\vec{w}|^2}{2} f^{(1)} \,d\vec{v} + \mathcal{O}(\mathrm{Kn}^2)} \\
    &=\rev{\frac{3R_s}{2e} \int_{\mathbb{R}^3} m\, \vec{w} \frac{|\vec{w}|^2}{2} \left(f^{(0)} + \mathrm{Kn}\, f^{(1)}\right) \,d\vec{v}} \\
    &= \rev{\frac{3R_s}{2e} \vec{Q}^{(1)}} + \mathcal{O}(\mathrm{Kn}^2),
\end{align}
\end{subequations}
where \rev{we employed \eqref{TTT3} and again used} the solvability condition \eqref{eq:solvability}.
\end{proof}

\subsection{Constitutive equation for entropy}
\rev{Using \eqref{TTT}}, we can compute the entropy per unit mass at equilibrium $\eta^{(0)}(\vec{x},t)$ using only $f^{(0)}(\vec{x},\vec{v},t)$, i.e.
\begin{equation}
	\label{eq:entropy0}
	\rho \eta^{(0)} = \int_{\mathbb{R}^3}\uze{-} {k_B}f^{(0)}(\vec{x},\vec{v},t) \log f^{(0)}(\vec{x},\vec{v},t) d\vec{v} = {R_s \rho \left[ \frac{3}{2} \log\left( {\frac{2}{3}}e \rho^{-\frac{2}{3}} \right) + C\right]},
\end{equation}
{where $C = \frac{3}{2}(\log{2\pi}+1) + \log m$} and we recall that $R_s$ is the specific gas constant.
Thus we have that $\eta^{(1)}=\eta^{(0)}=\eta^{(0)}(\rho,e)$ is a caloric equation of state, i.e.~it depends only on the macroscopic variables $\rho$ and $e$, and we can invert the relation to obtain that $e = e(\rho,\eta^{(0)})=e(\rho,\eta^{(1)})$. In particular, dropping $\mathcal{O}(\mathrm{Kn}^2)$ terms and redefining $C$, we have that
\begin{equation}
    \label{eq:eos0}
    e(\rho,\eta^{(0)}) = {\frac{3}{2}} \rho^{\frac{2}{3}} \exp\left( \frac{2}{3 R_s} \left( {\eta^{(0)}}-C \right) \right) = {\frac{3}{2}} \rho^{\frac{2}{3}} \exp\left( \frac{{\color{black}2}}{3 R_s} \left( {\eta^{(1)}}-C \right) \right) = e(\rho,\eta^{(1)}).
\end{equation}

\subsection{Thermodynamic temperature and caloric equation of state}
The thermodynamic definition of temperature $\theta(\vec{x},t)$ in Kelvin is given by
\begin{equation}
    {\theta \coloneqq \frac{\partial e}{\partial \eta^{(0)}}}
\end{equation}
where once again we have dropped $\mathcal{O}(\mathrm{Kn}^2)$ terms.
The temperature $T$ in units of energy per mass is related to the temperature $\theta$ in Kelvin by the specific gas constant ${R_s}$:
\begin{equation}
T = {R_s} \theta.
\end{equation}
Using \eqref{eq:eos0} we can derive the equipartition of energy \eqref{eq:equipartition_of_energy}. Differentiating the internal energy with respect to entropy yields
\begin{equation} \label{eq:thermal_equation_of_state}
    \theta = \frac{{4}}{9 R_s}\rho^{\frac{2}{3}}\exp\left(\frac{2}{3 R_s}({\eta^{(0)}}-C)\right)=\frac{2}{3 R_s}e \Rightarrow e = \frac{3}{2} T = \frac{3}{2}{R_s}\theta.
\end{equation}
This result, i.e.~$e=\frac{3}{2}R_s\theta$, is the caloric equation of state for a monatomic ideal gas.
The equipartition of energy theorem for monatomic gases can thus be derived as a consequence of the Chapman--Enskog expansion.

\subsection{Thermodynamic pressure and the thermal equation of state (ideal gas law)}
The thermodynamic definition of pressure $p(\vec{x},t)$ is given by
\begin{equation}
    \label{eq:termopressure}
    p \coloneqq \rho^2 \frac{\partial e(\rho, \eta^{(0)})}{\partial \rho}.
\end{equation}
Using \eqref{eq:eos0} and \eqref{eq:thermal_equation_of_state} we can express the pressure in terms of $\rho$ and $\theta$, as
\begin{equation}\label{eq:idealgas}
    p = \rho^2 \frac{\partial}{\partial \rho}\left[\frac{{3}}{2} \rho^{\frac{2}{3}} \exp\left( \frac{2}{3 R_s} \left( \eta^{(0)}-C \right) \right)\right] = \frac{2}{3} e\rho = {R_s}\theta \rho.
\end{equation}
This last formula is the well known ideal gas law. 
\begin{remark}
    Since $\rho$ is expansion-independent, from \eqref{eq:termopressure} we see that also the pressure is expansion-independent.
    This fact is reconfirmed by the above calculation where the pressure is expressed in terms of $\rho$ and $\theta$, both of which are expansion-independent, and thus the pressure is expansion-independent as well.
\end{remark}

Rajagopal \cite{rajagopalPressure} observed that the term pressure is misused in the literature, in that it is ambiguously used to refer both to the thermodynamic pressure and to the mechanical pressure.
While the thermodynamic pressure is defined in \eqref{eq:termopressure}, the mechanical pressure $\hat{p}$ is defined as the spherical component of the stress tensor:
\begin{equation}
    \mathbb{T} = \mathbb{T}^{d}+\hat{p}\mathbb{I},
\end{equation}
where we used the superscript $d$ to indicate the deviatoric part of the Cauchy stress tensor, i.e.~$\cdot^{d}\coloneqq \cdot - {\frac{1}{3}}\text{tr}(\cdot)\mathbb{I}$, and thus
\begin{equation}
    \hat{p}= \frac{1}{3}\text{tr}(\mathbb{T}).
\end{equation}
A well known limitation of the kinetic theory of rarefied gases is precisely the fact that the kinetic pressure and the mechanical pressure coincide up to sign \cite[Chapter 4]{Truesdell}, i.e.~$p=-\hat{p}$. Indeed, using the definitions \eqref{eq:newj3} and \eqref{def:energy} of $\mathbb{T}$ and $e$ together with the state equation \eqref{eq:idealgas}, we observe that
\begin{align}
	\label{eq:stokes_hypothesis_calculation}
	\hat{p} = \frac{\text{tr}}{3}\left[{-}\int_{\mathbb{R}^3} m\vec{w}\otimes\vec{w}f\,d\vec{v}\right] &= {-}\frac{1}{3}\int_{\mathbb{R}^3} m\,\text{tr}(\vec{w}\otimes\vec{w})f\,d\vec{v}\\
    &={-}\frac{1}{3} \int_{\mathbb{R}^3}m\,\abs{\vec{w}}^2 f\,d\vec{v} = {-}\frac{2}{3}\rho e = {-}p.
\end{align}
This observation is often referred to as the Stokes hypothesis and in continuum mechanics it is often expressed by introducing the concept of bulk viscosity \cite{rajagopalPressure, Stokes}.
If we consider the general class of {compressible} Newtonian fluids where the stress tensor is given by
\begin{equation}
    \mathbb{T} = {-p\mathbb{I}} + \nu \left( \nabla_{\vec{x}} \vec{u} + (\nabla_{\vec{x}} \vec{u})^T - \frac{2}{3} (\nabla_{\vec{x}} \cdot \vec{u}) \mathbb{I} \right) + \zeta (\nabla_{\vec{x}} \cdot \vec{u}) \mathbb{I},
\end{equation}
where $\nu \ge 0$ is the dynamic viscosity and $\zeta \ge 0$ is the bulk viscosity, then the identity $p=-\hat{p}$ holds if $\zeta=0$.

\begin{remark}
	\rev{
        In the case of monatomic ideal gases considered here, we have just shown that the kinetic pressure and the mechanical pressure coincide up to sign. For more complex fluids, such as polyatomic gases, the situation is more involved. For example, Curtiss' model for polyatomic gases \cite{curtissI}, which will be briefly discussed in \Cref{sec:LE}, accounts for both translational and rotational degrees of freedom of top-symmetric molecules representing the polyatomic gas constituents. In this model the total internal energy splits as $e = e_{tr} + e_{rot}$, where the translational and rotational contributions satisfy $e_{tr} = \frac{3}{2}R_s\theta_{tr}$ and $e_{rot} = \frac{N-3}{2}R_s\theta_{rot}$ respectively, with $\theta_{tr}$ and $\theta_{rot}$ the corresponding partial temperatures and $N$ the total number of active degrees of freedom.
        At thermodynamic equilibrium both $\theta_{tr}$ and $\theta_{rot}$ coincide with the single thermodynamic temperature $\theta$, and the ideal gas law $p = \rho R_s \theta$ holds unchanged.
	The mechanical pressure, defined as $\hat{p} = \frac{1}{3}\mathrm{tr}(\mathbb{T})$, depends only on $\vec{w}\otimes\vec{w}$ and is given by $\hat{p} = -\frac{2}{3}\rho e_{tr}$.
        At equilibrium $\theta_{tr} = \theta$ and hence $\hat{p} = -\rho R_s\theta = -p$, so kinetic and mechanical pressures still coincide up to sign.
        Out of equilibrium, however, $\theta_{tr} \neq \theta$ and consequently $\hat{p} \neq -p$, the discrepancy being controlled by the bulk viscosity $\zeta$.}

        \rev{It is important to stress that it is not merely the existence of internal degrees of freedom that leads to a non-vanishing bulk viscosity, but rather the fact that the rotational degrees of freedom relax to equilibrium at a different rate than the translational ones, creating a lag in the equilibration of the two partial temperatures. This is what generates a non-zero bulk viscosity $\zeta$. 
	This mechanism becomes transparent in the BGK approximation for polyatomic gases of Andries, Aoki and Perthame \cite{PerthameBGK}, where the following limiting scenarios can be identified:}
	\begin{enumerate}
            \item \rev{If the internal degrees of freedom and translational degrees of freedom relax at the same rate, i.e.~$\tau_{int} = \tau_{tr}$, the model reduces to a single-species BGK equation, giving $\zeta = 0$ exactly, and thermodynamical and mechanical pressures coincide up to sign.}
            \item \rev{If $\tau_{int} > \tau_{tr}$, i.e.~rotational modes relax more slowly as is the case for physical polyatomic gases, then $\zeta > 0$ and kinetic and mechanical pressure no longer coincide up to sign.}
                 \rev{The larger the ratio $\tau_{int}/\tau_{tr}$, the more pronounced the bulk viscosity.}	
        \end{enumerate}
\end{remark}

\subsection{Zeroth and first order entropy flux and entropy production: Clausius--Duhem inequality}
\begin{theorem}
    \label{prop:RSConstraint}
    Under the assumption that the Chapman--Enskog expansion can be truncated at first order, i.e.,
    \begin{equation}\label{eq:newj12} 
    f(\vec{x},\vec{v},t) \approx f^{(0)}(\vec{x},\vec{v},t) + \mathrm{Kn} \, f^{(1)}(\vec{x},\vec{v},t)
    \end{equation} 
    then, for $k=0,1$, 
    \begin{align}\label{eq:newj13}
        \rho \left(\frac{\partial \eta^{(k)}}{\partial t} + \vec{u} \cdot \nabla_{\vec{x}}  \eta^{(k)} \right) + \nabla_{\vec{x}} \cdot \vec{\Phi}^{(k)} = \xi^{(k)},
    \end{align}
    where
    \begin{equation}
        \label{eq:newj14}
        {\vec{\Phi}}^{(k)} = \frac{\vec{Q}{^{(k)}}}{{\theta}} \quad \textrm{ and } \quad {\theta} \xi^{(k)} = ({\mathbb{T}^{(k)}} + p \mathbb{I}): \mathbb{D} - \vec{Q}{^{(k)}} \cdot \frac{\nabla_{\vec{x}} {\theta}}{{\theta}} ,
    \end{equation}
{with $\mathbb{D}$ denoting the symmetric part of $\nabla_{\vec{x}}\vec{u}$}.
\end{theorem}

\begin{proof}
Under the assumption \eqref{eq:newj12}, it follows from \eqref{eq:eos0} that $e = e(\rho, \eta^{(k)})$, $k=0,1$. Applying the material time derivative $\frac{d\cdot}{dt}:= \frac{\partial \cdot}{\partial t} + \vec{u}\cdot \nabla_{\vec{x}} \cdot$ to both sides of these equations we obtain 
\begin{equation}
	\frac{de}{dt}  = \frac{d}{dt} e(\rho,\eta^{(k)}) = \frac{\partial e}{\partial \rho} \frac{d\rho}{dt} + \frac{\partial e}{\partial \eta^{(k)}} \frac{d\eta^{(k)}}{dt}= \frac{p}{\rho^2}\frac{d\rho}{dt} {+} {\theta} \frac{d\eta^{(k)}}{dt},\qquad k=0,1,
\end{equation}
where we have used the thermodynamic definition of temperature and pressure.  A simple manipulation followed by the use of balance equations \eqref{eq:evolution_internal_energy} and \eqref{eq:balance_law_mass} leads to 
\begin{align}
	{\theta} \rho \frac{d\eta^{(k)}}{dt} = \rho \frac{de}{dt}  - \frac{p}{\rho}\frac{d\rho}{dt} &= {\mathbb{T}^{(k)}}: \nabla_{\vec{x}} \vec{u} - \nabla_{\vec{x}} \cdot {\vec{Q}^{(k)}} + p \nabla_{\vec{x}} \cdot\vec{u} \\ &= ({\mathbb{T}^{(k)}} + p \mathbb{I}): \nabla_{\vec{x}} \vec{u} - \nabla_{\vec{x}} \cdot {\vec{Q}^{(k)}},\qquad k=0,1. 
\end{align}
After multiplying by ${\frac{1}{\theta}}$ and rearranging the terms, we get 
\begin{align}
	\rho \frac{d\eta^{(k)}}{dt} + \nabla_{\vec{x}} \cdot \left(\frac{\vec{Q}{^{(k)}}}{{\theta}}\right) = \frac{1}{{\theta}} \left[({\mathbb{T}^{(k)}} + p \mathbb{I}): \nabla_{\vec{x}} \vec{u} - \vec{Q}{^{(k)}} \cdot \frac{{\nabla_{\vec{x}} \theta}}{{\theta}} \right],\qquad k=0,1, 
\end{align}
which gives the assertion {(after employing the symmetry of $\mathbb{T}$)}. 
\end{proof}
\begin{remark}
	We have thus shown that sufficiently close to thermodynamic equilibrium, i.e.~when the Chapman--Enskog expansion can be truncated at the zeroth or first order, the entropy flux vector $\vec{\Phi}^{(k)}$ reduces to the classical form from continuum thermodynamics $\vec{\Phi}^{(k)} = \vec{Q}{^{(k)}}/{\theta}$,and thus the entropy balance law \eqref{eq:almostClausiusDuhem} reduces to the classical Clausius--Duhem inequality.
    This also suggests that far from thermodynamic equilibrium the entropy flux vector might need be modified to take into account non-equilibrium effects.
\end{remark}
\section{Euler and Navier--Stokes--Fourier equations}
By considering zero and first order approximations of the Chapman--Enskog expansion respectively, we now proceed to derive the constitutive relations for Euler and Navier--Stokes--Fourier fluids.
The procedure outlined here is different from the classical Chapman--Enskog procedure, yet obtains the same constitutive relations.
In \Cref{sec:LE} we will present a situation where the procedure outlined here and the classical Chapman--Enskog procedure yield different results.
\subsection{Zeroth order approximation: Euler equations}
Considering only the zeroth-order approximation of the Chapman--Enskog expansion, i.e.,
\begin{equation}
    \label{eq:zerothOrderCE}
    f(\vec{x},\vec{v},t) \approx f^{(0)}(\vec{x},\vec{v},t),
\end{equation}
then $\xi^{(0)}\equiv 0$, which trivially follows from \eqref{eq:BGK_entropy_production}, by substituting $f^{(0)}$ instead of $f$.
It then follows from the second equation in \eqref{eq:newj14}, that 
\begin{equation}\label{euler}
            \mathbb{T}{^{(0)}} = - {p} \mathbb{I},\quad \vec{Q}{^{(0)}} = 0,
\end{equation}
is an admissible choice of constitutive relations for the Cauchy stress tensor and the heat flux vector, recovering the well-known constitutive relations for compressible Euler (ideal) fluids.

However one can construct many different constitutive relations for $\mathbb{T}{^{(0)}}$ and $\vec{Q}{^{(0)}}$ that satisfy the second equation in \eqref{eq:newj14} with $\xi^{(0)}\equiv 0$.
For example one can consider
\begin{equation}
    \label{eq:counterexample}
    \mathbb{T}{^{(0)}} = - {p} \mathbb{I} - \alpha \nabla_{\vec{x}} \theta \otimes \nabla_{\vec{x}} \theta, \quad \vec{Q}{^{(0)}} = \rev{-}\alpha \rev{\theta}\mathbb{D}\nabla_{\vec{x}} \theta,
\end{equation}
where $\alpha \in \mathbb{R}$ and $\mathbb{D}$ is the symmetric part of $\nabla_{\vec{x}} \vec{u}$.
Thus a selection criterion has to be invoked to select the constitutive relations \eqref{euler} over the ones in \eqref{eq:counterexample}.

\begin{criterion}[\rev{\textbf{Linear irreversible thermodynamic procedure}}]
    \label{selec:NC}
    Under the \rev{linear irreversible thermodynamic} assumption that the constitutive relations are piecewise linear relationships between thermodynamic fluxes $\{\mathbb{T}+p\mathbb{I}, \vec{Q}\}$ and thermodynamic affinities $\{\rev{\mathbb{D}^d},\rev{-}\nabla_{\vec{x}} {\theta}/{\theta}\}$ \cite{deGrootMazur}, we can invoke the rational thermodynamics idea \cite{RationalThermodynamics} that \eqref{eq:newj14} should hold for any choice of affinities, and thus we can uniquely select the constitutive relations \eqref{euler}.
\end{criterion}

\begin{criterion}[\textbf{Classical Chapman--Enskog procedure}]
\label{selec:CE}
At this level, it is worth observing that \eqref{euler} can be obtained directly by inserting $f^{(0)}$ in the definition of $\vec{Q}$ and $\mathbb{T}$, as per the classical Chapman--Enskog procedure.
Indeed, using the properties of the Maxwellian and the fact that the mechanical pressure is equal to the thermodynamical pressure, we obtain 
\begin{equation}
    \vec{Q}{^{(0)}} = \int_{\mathbb{R}^3} \frac{1}{2} m \abs{\vec{w}}^2 \vec{w} f^{(0)}(\vec{x},\vec{v},t) \,d\vec{w} = 0 \implies  \vec{\Phi}{^{(0)}}(\vec{x},t) = 0,
\end{equation}
by a symmetry argument, since the Maxwellian distribution is an even function of $\vec{w}$, and
\begin{align}
    \mathbb{T}{^{(0)}} &\!={-}\int_{\mathbb{R}^3}\!\!m \vec{w} \otimes \vec{w} f^{(0)} \,d\vec{w}\!=\!\! {-}\int_{\mathbb{R}^3}\!\! m \frac{1}{3} \abs{\vec{w}}^2 f^{(0)} {\mathbb{I}}\,d\vec{w} \!=\! 
    -p \mathbb{I},
\end{align}
where the last identity comes from the fact that $\mathbb{E}[\vec{w}\otimes \vec{w}]$ is the variance of the Maxwellian distribution which, since \eqref{eq:maxwellian} is a normal distribution with diagonal covariance matrix, can be computed as $\mathbb{E}[\frac{1}{3}\abs{\vec{w}}^2 \mathbb{I}]$.
\end{criterion}
\begin{remark}
    \label{selec:KT}
    One can adopt a hybrid procedure where information from the kinetic theory is used to restrict the class of admissible constitutive relations. 
    For example, if we limit ourselves to constitutive relations of the form of \eqref{eq:counterexample}, the pressure could also play the role of a selection criterion.
    In fact starting from the thermodynamic definition of pressure \eqref{eq:termopressure} together with the equation of state \eqref{eq:eos0} we have that the mechanical pressure $\hat{p}$ and the thermodynamic pressure $p$ coincide up to second order in Knudsen number.
    Thus using \eqref{eq:counterexample} we have that 
    \begin{equation}
            {-\frac{\alpha}{3} \abs{\nabla {\theta}}^2 = \hat{p} + p =  \mathcal{O}(\mathrm{Kn}^2)},
    \end{equation}
    implying that $\alpha$ must be of order $\mathcal{O}(\mathrm{Kn}^2)$ and so the alternative constitutive closure \eqref{eq:counterexample} reduces to \eqref{euler} up to first order in Knudsen number.
    We remark that the pressure based selection criterion is less general than other selection criteria we will present and it relies on the Stokes hypothesis, i.e.~that there is no bulk viscosity $\zeta$.
    For example, if we consider the alternative constitutive closure
    \begin{equation}
	    \mathbb{T}{^{(0)}} = - {p} \mathbb{I} - \alpha \nabla_{\vec{x}} \theta \otimes \nabla_{\vec{x}} \theta + \frac{\alpha}{3}\abs{\nabla \theta}^2 \mathbb{I}, \quad \vec{Q}{^{(0)}} = \rev{-}\alpha \rev{\theta}\mathbb{D}^d \nabla_{\vec{x}} \theta,
    \end{equation}
    with $\mathbb{D}^d$ the deviatoric part of $\mathbb{D}$, then the pressure no longer plays a role in the selection criterion.
\end{remark}

\subsection{First order approximation: Navier--Stokes equations}
\label{sec:FirstOrder}
We continue with the first order approximation of the Chapman--Enskog expansion \eqref{eq:newj12}.

To compute the entropy production rate at first order we begin by computing $f^{(1)}(\vec{x},\vec{v},t)$ using the classical Chapman--Enskog procedure, i.e.~we insert \eqref{eq:newj12} into \eqref{eq:BoltzmannMultipleScales} to obtain
\begin{align}
    \left(\partial_{t_0} + \mathrm{Kn} \,\partial_{t_1} \right) \left( f^{(0)} + \mathrm{Kn} \,f^{(1)} \right) + \vec{v} \cdot \nabla_{\vec{x}} \left( f^{(0)} + \mathrm{Kn} f^{(1)} \right) = -\frac{1}{\tilde{\tau}} f^{(1)},
\end{align}
where $\mathrm{Kn}\,\tilde{\tau} = \tau$ and $\tilde{\tau}$ is a rescaled relaxation time independent of the Knudsen number.
We then collect terms at different orders of the Knudsen number.
Starting from zeroth-order we have
\begin{equation}
    \partial_{t_0} f^{(0)} + \vec{v} \cdot \nabla_{\vec{x}} f^{(0)} = -\frac{1}{\tilde{\tau}} f^{(1)},
\end{equation}
which can be inverted to obtain $f^{(1)}(\vec{x},\vec{v},t)$ as
\begin{equation}
    \label{eq:ce_starting}
    f^{(1)}(\vec{x},\vec{v},t) = -\tilde{\tau} \left( \partial_{t_0} f^{(0)} + \vec{v} \cdot \nabla_{\vec{x}} f^{(0)} \right).
\end{equation}
Together with the solvability condition \eqref{eq:solvability} this relation allows us to compute $f^{(1)}(\vec{x},\vec{v},t)$ explicitly in terms of the macroscopic quantities $\rho(\vec{x},t)$, $\vec{u}(\vec{x},t)$, and $e(\vec{x},t)$ and their spatial derivatives \cite{Chapman,Harris,cercignani}, i.e.
\begin{equation}
    \label{eq:first_order_correction}
    f^{(1)}(\vec{x},\vec{v},t) = -\tilde{\tau} f^{(0)}\left[
    \frac{3}{2}\frac{(\vec{w}\otimes \vec{w})^d}{e}:\mathbb{D}^d \rev{+} \left(\frac{3}{4}{\frac{\abs{\vec{w}}^2}{e}} - \frac{5}{2} \right) \vec{w}\cdot \frac{\nabla_{\vec{x}} e}{e} 
    \right]. 
\end{equation}
While the detailed calculations of how to compute $f^{(1)}$ are standard in the literature \cite{cercignani,Chapman,Harris}, for the sake of completeness they are presented in \ref{appendix_CE}.
Using this expression for $f^{(1)}(\vec{x},\vec{v},t)$ we can compute the entropy production rate at first order as
\begin{align}
    \xi^{(1)}(\vec{x},t) &= \frac{ \rev{\mathrm{Kn}}}{\tau}\int_{\mathbb{R}^3} {k_B} f^{(1)}\log\left( f^{(0)} + \mathrm{Kn} f^{(1)}\right) \,d\vec{v} \\
    &\approx \frac{\mathrm{Kn}}{\tau}\int_{\mathbb{R}^3}{k_B} f^{(1)} \left[ \log(f^{(0)}) + \mathrm{Kn} \frac{f^{(1)}}{f^{(0)}} \right] d\vec{v} + \mathcal{O}(\mathrm{Kn}^3) \label{eq:aux}\\
    &{= \frac{\mathrm{Kn}^2}{\tau} \int_{\mathbb{R}^3}k_B \frac{(f^{(1)})^2}{f^{(0)}} d\vec{v} + \mathcal{O}(\mathrm{Kn}^3)}\label{eq:aux1}\\
    &={\tau} \int_{\mathbb{R}^3} {k_B}\left[\frac{3}{2}\frac{(\vec{w}\otimes \vec{w})^d}{e}:\mathbb{D}^d \right]^2f^{(0)} d\vec{v} \nonumber \\
    &+2{\tau}\int_{\mathbb{R}^3} {k_B}\left[\frac{3}{2}\frac{(\vec{w}\otimes \vec{w})^d}{e}:\mathbb{D}^d \right]\cdot \left[\left(\frac{3}{4}\frac{\abs{\vec{w}}^2}{e} - {\frac{5}{2}} \right) \vec{w}\cdot \frac{\nabla_{\vec{x}} e}{e}\right] f^{(0)} d\vec{v} \nonumber \\
    &+{\tau} \int_{\mathbb{R}^3} {k_B}\left[\left(\frac{3}{4}\frac{\abs{\vec{w}}^2}{e} - \frac{5}{2} \right) \vec{w}\cdot \frac{\nabla_{\vec{x}} e}{e}\right]^2 f^{(0)} d\vec{v} + \mathcal{O}(\mathrm{Kn}^3),
\end{align}
{where we employed the fact that $\int_{\mathbb{R}^3}f^{(1)}\log{f^{(0)}}\,d\vec{v} = 0$ as a consequence of {\eqref{eq:solvability}}.} 

We can simplify the above expression by introducing the viscosity $\nu$ and heat conductivity $\kappa$ to obtain the following compact expression for the entropy production at first order. Details can be found in \ref{appendix_CE} and are standard in the literature \cite{cercignani,Chapman,Harris}:
\begin{align}
       {\theta} \xi^{(1)}(\vec{x},t)    &={2\nu (\mathbb{D}^d : \mathbb{D}^d) + \kappa\frac{\abs{\nabla_{\vec{x}}{\theta}}^2}{{\theta}}} + \mathcal{O}(\mathrm{Kn}^3),
    \label{eq:entropyProductionFirstOrder}
\end{align}
where the viscosity and heat conductivity can be explicitly computed as
\begin{equation}
    \nu = \tau p, \quad \kappa = \frac{5}{2} R_s \tau p.
\end{equation}
The entropy production expression \eqref{eq:entropyProductionFirstOrder} must hold along with the constraint \eqref{eq:newj14}, i.e.
\begin{equation}
\label{eq:auxo2}
{\theta} \xi^{(1)} =   (\mathbb{T}^{(1)} + p\mathbb{I}) : {\mathbb{D}} {-} \frac{\vec{Q}^{(1)}\cdot\nabla_{\vec{x}}{\theta}}{{\theta}}.
\end{equation}
Once again we must invoke a selection criterion to select the constitutive relations for $\mathbb{T}^{(1)}$ and $\vec{Q}^{(1)}$ that satisfy the above relation.

Invoking again \Cref{selec:NC} for the two previous relations \eqref{eq:entropyProductionFirstOrder} and \eqref{eq:auxo2} to hold for arbitrary velocity and temperature fields, assuming piecewise linear relationship between thermodynamic fluxes $\{\mathbb{T}+p\mathbb{I}, \vec{Q}\}$ and thermodynamic affinities $\{\rev{\mathbb{D}^d},\rev{-}\nabla_{\vec{x}}{\theta}/{\theta}\}$, 
we must have up to first order in Knudsen number
\begin{equation}
\label{eq:ns-closures}
    \mathbb{T}{^{(1)}} = -p \mathbb{I} \,{+}\, 2 \nu \mathbb{D}^d, \quad \vec{Q}{^{(1)}} = -{\kappa \nabla_{\vec{x}} {\theta}},
\end{equation}
thus recovering the well-known constitutive relations for a compressible Newtonian fluid (with zero bulk viscosity) and for Fourier heat-conduction.

\begin{remark}
Note that, outside the realm of linear constitutive relationships, the constitutive closures for the irreversible part of the Cauchy stress, $\mathbb{T}+p\mathbb{I}$, and for the heat flux vector $\vec Q$ can no longer be uniquely determined from \eqref{eq:auxo2} and \eqref{eq:entropyProductionFirstOrder}, as before.
Analogously to \eqref{eq:counterexample} we can construct infinitely many different constitutive closures of the form of
\begin{equation}
\label{eq:alternative_closure}
\mathbb{T}{^{(1)}}+p\mathbb{I} = 2\nu \mathbb{D}^d - \alpha \nabla {\theta} \otimes \nabla {\theta},
\qquad
\vec Q{^{(1)}} = -\kappa \nabla {\theta} - \alpha {\theta} \mathbb{D}\nabla {\theta},
\end{equation}
for $\alpha$ arbitrary, which represents a possible alternative constitutive closure that is no longer linear in the thermodynamic affinities.
\end{remark}
Another selection criterion, which we consider of more fundamental nature, can be provided by the principle of maximum entropy production.
\begin{criterion}
    Rajagopal and Srinivasa \cite{rajagopalSrinivasa} introduced an alternative selection criterion, which allows one to uniquely determine the constitutive relations for the thermodynamic fluxes, here $\mathbb{T}+p\mathbb{I}$ and $\vec Q$, from the two expressions for entropy production given in \eqref{eq:entropyProductionFirstOrder} and \eqref{eq:auxo2}.
    This criterion states that the constitutive closures for the thermodynamic fluxes are those for which the entropy production \eqref{eq:entropyProductionFirstOrder} is maximized with respect to the thermodynamic affinities, here $\rev{\mathbb{D}^d}$ and $\rev{-}\nabla_{\vec x} {\theta} / {\theta}$, subject to the linear constraint \eqref{eq:auxo2}.
    \rev{Equivalently, via a duality argument \cite{vitMalek}, the entropy production is maximized over the thermodynamic fluxes, here $\mathbb{T}+p\mathbb{I}$ and $\vec Q$, subject to the same linear constraint \eqref{eq:auxo2} (see \Cref{principle:RS} for a general statement).} Solving this constrained optimization problem indeed uniquely recovers the classical closures \eqref{eq:ns-closures}.
\end{criterion}
\rev{The dual formulation of the principle of maximum entropy production fits particularly well with the setting of implicit constitutive relations \cite{RajagopalImplicit}.} We refer to the combination of this selection procedure together with the computation of constitutive relation for entropy and entropy production via the Chapman--Enskog expansion as a \textit{hybrid Rajagopal--Srinivasa--Chapman--Enskog} closure.
\begin{remark}
    The Stokes hypothesis, i.e.~the assumption that the bulk viscosity is zero, which must hold as discussed in \Cref{sec:EoS} for the kinetic theory of rarefied gases, can also be seen as a consequence of the fact that the entropy production rate at first order \eqref{eq:entropyProductionFirstOrder} only depends on the deviatoric part of the tensor $\mathbb{D}^d$, and not on the volumetric part $\frac{1}{3}(\nabla_{\vec{x}} \cdot \vec{u})\mathbb{I}$.
    In analogy to \Cref{selec:KT}, the Stokes hypothesis can also be used to restrict the class of admissible constitutive relations. 
\end{remark}

Notice that Rajagopal and Srinivasa's selection criterion does not select between \eqref{euler} and \eqref{eq:counterexample} at the level of the zeroth-order approximation. However, we can still apply it working backwards from the first order expansion.
As the Knudsen number goes to zero, we expect the first order Chapman--Enskog expansion $f=f^{(0)} + \mathrm{Kn} f^{(1)}$ to converge to the zeroth-order expansion $f=f^{(0)}$.
Thus, we expect that the first order entropy production $\xi^{(1)}$ converges to the zeroth-order entropy production $\xi^{(0)} \equiv 0$ and the first order constitutive relations \eqref{eq:ns-closures} to converge to the Euler constitutive relations \eqref{euler}.
By this argument we identify \eqref{euler} as the unique limit of the first order constitutive relations \eqref{eq:ns-closures} as the Knudsen number goes to zero. This is confirmed by the fact that both the viscosity $\nu$ and the heat conductivity $\kappa$ scale like $\mathrm{Kn}$ in view of \eqref{eq:tau_Kn}.

This idea can be generalised to any of the selection criteria.
\begin{criterion}[\textbf{Backward compatibility procedure}]
    \label{selec:BK}
    Given a selection criterion that allows one to select the constitutive relations for the thermodynamic fluxes at first order, we can select the zeroth-order constitutive relations as the unique limit of the first order constitutive relations as the Knudsen number goes to zero.
\end{criterion}

\section{Relaxation time and maximization of the rate of entropy production principles}
In the previous section we recovered the traditional thermodynamic framework of
linear irreversible thermodynamics in the setting of Navier--Stokes--Fourier
fluids. By invoking the standard thermodynamic concepts of \emph{thermodynamic
fluxes}---the rates at which a physical quantity such as mass, charge, or
energy is transported through the system---and \emph{thermodynamic affinities},
which represent the forces driving this transport, we obtained two expressions
for the entropy production,
\begin{equation}
\theta\xi = \widehat{\xi}(\mathbf{A}), \qquad \theta\xi = \mathbf{J}\cdot\mathbf{A},
\end{equation}
corresponding respectively to \eqref{eq:entropyProductionFirstOrder} and \eqref{eq:auxo2}.
In this formulation, the thermodynamic fluxes correspond
to the components of $\mathbb{T}+p\mathbb{I}$ and $\mathbf{Q}$, while the
corresponding affinities are the components of $\rev{\mathbb{D}^d}$ and $\rev{-}\nabla \theta / \theta$.

We now propose an extension of the procedure outlined
above. This extension is guided partly by theory and partly by empirical considerations.
Under the assumptions of a first-order Chapman--Enskog truncation and a BGK collision operator, we arrived at (see \rev{\eqref{eq:aux1}})
\begin{align}
      \label{eq:intermidiate}
\xi^{(1)}(\vec{x},t)
      & = \frac{\rev{k_B}}{\tau}\int_{\mathbb{R}^3} 
        \mathrm{Kn}^2f^{(1)}
        \left[ 
            \frac{f^{(1)}}{f^{(0)}}
        \right]
      d\vec{v}
      + \mathcal{O}(\mathrm{Kn}^3).
\end{align}
At this point we assume that substituting $f^{(1)}$, whose form depends on the collision operator, into the previous expression, yields
\begin{equation}
    \label{eq:aux2}
	\theta\xi^{(1)}(\vec{x},t) = \frac{1}{\tau} \,\iota(\mathbf{A}(\vec{x},t)),
\end{equation}
where $\iota$ is a strictly convex function of the affinities $\mathbf{A}$.
We will discuss the validity of this assumption for the particular case of the Boltzmann equation with the BGK collision operator in \Cref{ex:whenDoesItHold}.

The relaxation time reflects the speed of local equilibration and thus
represents an essential characteristic of the material response. Motivated by
this interpretation, {\it we suppose that the relaxation time
is determined by the thermodynamic fluxes and affinities}.
In this spirit, we generalize \eqref{eq:aux2} by introducing the ansatz
\begin{equation}
\label{eq:thermodynamic-ansatz}
\theta\xi(\vec{x},t)
   = {\frac{1}{\hat{\tau}({\mathbf{J}}(\vec{x},t),\vec{A}(\vec{x},t))}
     \,\iota(\mathbf{A}(\vec{x},t))},
\end{equation}
where $\iota(\mathbf{A})$ is computable from kinetic theory, together with the independent constraint arising from the balance laws,
\begin{equation}
    \theta\xi(\vec{x},t) = \mathbf{J}(\vec{x},t)\cdot\mathbf{A}(\vec{x},t),
\end{equation}
and the assumption
\begin{equation}
    \tau(\vec{x},t) = {\hat{\tau}({\mathbf{J}}(\vec{x},t),\vec{A}(\vec{x},t))}.
\end{equation}
Our task is to identify a unique relation between thermodynamic fluxes and affinities.
This is precisely the situation in which
Rajagopal and Srinivasa~\cite{rajagopalSrinivasa} propose the additional
selection criterion of \emph{maximization of the rate of entropy production}.
In the present context, this principle may be formulated as follows.
\begin{principle}[\textit{Rajagopal--Srinivasa entropy production maximization}]
\label{principle:RS}
If a system described locally by fluxes $\mathbf{J}$ and affinities
$\mathbf{A}$ has entropy production $\theta \xi = \widehat{\xi}(\mathbf{A},\mathbf{J})$
such that $\widehat{\xi}(\mathbf{A},\cdot)$ is strictly convex and nonnegative for each
$\mathbf{A}$, and it is subject to the constraint $\theta\xi = \mathbf{J}\cdot\mathbf{A}$,
then the material response uniquely relating $\vec{J}$ and $\vec{A}$ is obtained from
the constrained maximization problem
\begin{equation}
\underset{\mathbf{J}}{\mathrm{maximize}}\;{\theta}\widehat{\xi}(\mathbf{J},\mathbf{A})
   \qquad\text{subject to}\qquad
   {\xi} = \frac{1}{\theta}\mathbf{J}\cdot\mathbf{A}
\end{equation}
\end{principle}

We will relate this to a principle of minimal relaxation time. 

\begin{principle}[Minimal relaxation time principle]
\label{principle:MinRelaxTime}
Let the local relaxation time $\tau$ for a process with entropy
production $\xi$ be fully determined by the thermodynamic fluxes $\mathbf{J}$.
Then the thermodynamic response, for a fixed set of affinities $\mathbf{A}$, is such as to \emph{minimize the relaxation
time towards equilibrium}.
\end{principle}

\begin{proposition}
\label{prop:maxEntropy}
Under the assumption that $\theta\widehat\xi(\mathbf{J},\mathbf{A})$ is a convex non-negative function of $\mathbf{J}$ for each $\mathbf{A}$, and 
\begin{equation}
    \label{eq:tau_iota_xi}
    \theta\widehat\xi(\mathbf{J},\mathbf{A}) = {\frac{1}{\tau({\mathbf{J}},\mathbf{A})^{\rev{\alpha}}}}\,\iota(\mathbf{A}),
\end{equation}
with $\rev{\alpha>0}$, $\iota$ a positive functional of the {affinities}, \Cref{principle:RS} is equivalent to 
\Cref{principle:MinRelaxTime}.
For a {prescribed} set of affinities
$\mathbf{A}$, the relaxation time is given
by
\begin{equation}
    \tau = \underset{{\mathbf{J}}\in\mathcal{J}}{\mathrm{min}}
      \;\tau({\mathbf{J}},\mathbf{A}),
\end{equation}
where the admissible constitutive relations are
$\mathcal{J}
   \coloneqq
   \left\{
       {\mathbf{J}\in \mathbb{R}^{d}}
       \,:\,
       \theta\xi = {\mathbf{J}}\cdot\mathbf{A}
   \right\}
$.
\end{proposition}

\begin{proof}
With $\mathbf{A}$ fixed, the quantity $\iota(\mathbf{A})$ is constant and positive.
	Thus by \eqref{eq:tau_iota_xi} minimizing $\tau(\mathbf{J}, \mathbf{A})$ is equivalent to maximizing $\hat{\xi}(\mathbf{J},\mathbf{A})$ as prescribed by \Cref{principle:RS}. 
\end{proof}
\begin{remark}
    Notice that as a result of the previous proposition, both the Rajagopal--Srinivasa principle and the minimal relaxation time principle produce a unique relation between thermodynamic fluxes and affinities.
    Furthermore such relations are by construction compatible with the second law of thermodynamics.
\end{remark}
\begin{example}
    \label{ex:whenDoesItHold}
        A natural question at this point is: when does the crucial assumption \eqref{eq:aux2} hold?
        Let us consider $f^{(1)}$ as in \eqref{eq:first_order_correction}, i.e.
        \begin{equation}
            f^{(1)}(\vec{x},\vec{v},t) = -\tilde{\tau} f^{(0)}\left[
                \frac{3}{2}\frac{(\vec{w}\otimes \vec{w})^d}{e}:\mathbb{D}^d + \left(\frac{3}{4}{\frac{\abs{\vec{w}}^2}{e}} - \frac{5}{2} \right) \vec{w}\cdot \frac{\nabla_{\vec{x}} e}{e} 
            \right], 
        \end{equation}
        which leads to the entropy production presented in \eqref{eq:entropyProductionFirstOrder}, i.e.
        \begin{equation}
		{\theta} \xi^{(1)}(\vec{x},t)={2 \tau p (\mathbb{D}^d : \mathbb{D}^d) + \frac{5}{2}  R_s\tau p \frac{\abs{\nabla_{\vec{x}}{\theta}}^2}{\theta}} + \mathcal{O}(\mathrm{Kn}^3).
        \end{equation}
        Expanding $p$ using \eqref{eq:idealgas} we obtain the following expression for the entropy production
        \begin{equation}
            \label{eq:auxo3}
            {\theta}\xi^{(1)}(\vec{x},t)={2\tau R_s{\theta} \rho (\mathbb{D}^d : \mathbb{D}^d) + \frac{5}{2}{R_s^2} \tau {\theta} \rho \frac{\abs{\nabla_{\vec{x}}{\theta}}^2}{{\theta}}} + \mathcal{O}(\mathrm{Kn}^3).
        \end{equation}
        It follows from \eqref{eq:scalingT} and \eqref{eq:thermal_equation_of_state} that 
        \begin{equation}
            \frac{1}{\tau} \propto \rho \sqrt{\theta} \implies \theta= \frac{1}{R_s}\left(C_{B}\rho\tau\right)^{-2},
        \end{equation}
        where $C_B$ is a constant depending on the molecular properties of the gas.
        Substituting the previous expression for $\theta$ into the right hand-side of \eqref{eq:auxo3}
        \begin{equation}
            {\theta}\xi^{(1)}(\vec{x},t)=\frac{1}{\tau}\left[\frac{2}{C_{\mathcal{B}}^2\rho}(\mathbb{D}^d:\mathbb{D}^d)+\frac{5}{2}\frac{{R_s}}{C_\mathcal{B}^2\rho}\frac{\abs{\nabla_{\vec{x}}{\theta}}^2}{{\theta}}\right] +\mathcal{O}(\mathrm{Kn}^3),
        \end{equation}
        which is exactly an entropy production in the form of \eqref{eq:aux2}.
\end{example}

\Cref{prop:maxEntropy} provides a kinetic interpretation of \Cref{principle:RS}, the
Rajagopal and Srinivasa selection procedure. Although the argument remains
phenomenological, it offers a new viewpoint: the material response is selected
as the one that yields the \emph{fastest possible relaxation} among all
admissible responses consistent with \eqref{eq:newj14}.

\section{Beyond Euler and Navier--Stokes: Leslie--Ericksen equations}
\label{sec:LE}
So far we have considered the Euler and Navier--Stokes--Fourier equations, where both the Chapman--Enskog and Chapman--Enskog--Rajagopal--Srinivasa procedures yield the same conclusions.
In this section, we briefly discuss a setting where the two approaches lead to different constitutive relations, drawn from \cite{farrell}.
We consider the Boltzmann--Curtiss equation \cite{curtissI,curtissV}, i.e.
\begin{equation}
	\partial_t f + \nabla_{\vec{x}}\cdot(\vec{v}f)+\nabla_{\vec{\alpha}}
    \cdot(\dot{\vec{\alpha}}f) = Q_{BC}[f,f],\label{eq:boltzmannCurtiss}
\end{equation}
where the configuration space has been enlarged to include $\vec{\alpha}$, the vector of the three Euler angles describing the orientations of the molecules. The one particle distribution function $f$ now depends on all the components of the enlarged phase space, i.e.~$f=f(\vec{x},\vec{v},\vec{\alpha},\dot{\vec{\alpha}},t)$.
This kinetic equation was introduced to model dilute gases of rigid, non-spherical molecules, for which rotational degrees of freedom play a fundamental role in the macroscopic dynamics.
The Boltzmann--Curtiss collision operator $Q_{BC}$ is defined as
\begin{equation} \label{eq:collisionop}
 Q_{BC}[f,f] \coloneqq   \int\!\!\!\!\int\!\!\!\!\int\!\!\!\!\int (f_*^\prime f^\prime-f_*f)(\vec{k}\cdot\vec{\mathfrak{g}})dSd\vec{k}d\vec{p}_2d\vec{\alpha}_2d\rev{\dot{\vec{\alpha}}}_2,
\end{equation}
where $\vec{\mathfrak{g}}$ is the relative velocity of the point of contact and $dS\mathrm{d}\vec{k}$ is the surface element of the excluded volume of the two interacting molecules.

Since the phase space has been enlarged it is possible to prove that there is one additional independent collision invariant beyond those of the classical Boltzmann equation. This is the angular momentum \cite{carrillo,curtissI}, given by $\psi = \mathbb{N}\,\vec{\omega}+\vec{q}\times m\vec{v}$.\footnote{The microscopic angular momentum is also a collision invariant for the classical Boltzmann equation, but it is not independent from the other four classical collision invariants. For this reason the Boltzmann equation yields a symmetric stress tensor and no additional balance equation for the macroscopic angular momentum.}
Here $\vec{\omega}$ is the angular velocity of the molecule, which can be linearly related to $\dot{\vec{\alpha}}$ via a matrix $\Xi(\vec{\alpha})$ and $\mathbb{N}$ is the inertia tensor of the molecule.
This results in an additional balance law governing the macroscopic intrinsic angular momentum $\vec{\mu}$:
\begin{align}
    \rho\Big[\partial_t \vec{\mu} + (\nabla_{\vec{x}}\vec{\mu})\vec{u}\Big] -\nabla\cdot\mathbb{M} = \mathrm{skew}(\mathbb{T}-\mathbb{T}^T), \\
    \vec{\mu}\coloneqq \int \mathbb{N}\vec{\omega} f\,d\vec{v}d\vec{\dot{\alpha}}d\vec{\alpha}, \quad
    \rev{\vec{\gamma}\coloneqq \int \vec{\omega} f\, d\vec{v}d\vec{\dot{\alpha}}d\vec{\alpha}}, \\
    \mathbb{M}\coloneqq -\int \vec{w}\otimes(\mathbb{N}\vec{\omega})f \,d\vec{v}d\vec{\dot{\alpha}}d\vec{\alpha},
\end{align}
where $\mathrm{skew}$ associates a vector to a skew symmetric matrix and \rev{$\vec{\gamma}$ is the macroscopic angular velocity}.
The microscopically conserved energy $\psi = \frac{1}{2} m \abs{\vec{v}}^2 + \frac{1}{2}\vec{\omega} \cdot \mathbb{N}\vec{\omega}$ now includes an angular contribution, leading to the following equation governing the internal energy:
\begin{equation}
    \label{eq:energy_balance_le}
    \partial_t (\rho e )+\nabla_{\vec{x}}\cdot (\rho \vec{u}e ) -\mathbb{T}:\nabla_{\vec{x}}\vec{u}-\mathbb{M}:\nabla_{\vec{x}}\vec{\mu}+\nabla_{\vec{x}}\cdot \vec{Q} = 0,
\end{equation}
where $A:B$ represents the double contraction of two tensors $A$ and $B$, i.e.~$A:B = \sum_{i,j} A_{ij}B_{ij}= \mathrm{tr}(A^TB)$.

{
        To obtain the zeroth-order constitutive relations, following a Chapman--Enskog procedure, we simply evaluate the macroscopic fluxes on the Maxwellian distribution $f^{(0)}$
	\begin{equation}
		f^{(0)}(\vec{\alpha},\vec{V},\vec{\Omega})=\frac{nQ\sin(\alpha_2)}{\int_{\mathbb{R}^3}Q\sin(\alpha_2)\,d\alpha_2}\frac{m^{\frac{3}{2}}(I_1I_2I_3)^{\frac{1}{2}}}{(2\pi k_B \theta)^3}\exp\Big[-m\frac{\abs{\vec{V}}^2}{2 k_B \theta}-\frac{\vec{\Omega}\cdot\mathbb{N}\vec{\Omega}}{2 k_B \theta}\Big],\label{eq:maxwellian_Curtiss}
	\end{equation}
	where $\vec{V} = \vec{v}-\vec{u}$ is the peculiar velocity, $\vec{\Omega} = \vec{\omega}-\vec{\gamma}$ is the peculiar angular velocity, and $I_1, I_2, I_3$ are the principal moments of inertia of the molecule.
        
	Using the structure of $f^{(0)}$ and the fact that the translational Maxwellian is isotropic in $\vec{w}$:
        \begin{equation}
		\label{eq:T0}
	\mathbb{T}^{(0)} = -\int m\,\vec{w}\otimes\vec{w}\, f^{(0)}\,d\vec{v}\,d\vec{\alpha}\,d\dot{\vec{\alpha}} = -p\,\mathbb{I},
        \end{equation}
	where $p = \rho\mathrm{R}_s\theta = \rho T$ is the thermodynamic pressure, since we can evaluate $\mathbb{E}_{f^{(0)}}[\vec{w}\otimes\vec{w}]$ to be $\frac{1}{3}\abs{\vec{w}}^2\mathbb{I}$ by the isotropy of the Gaussian. In particular, $\mathbb{T}^{(0)}$ is symmetric and isotropic, hence the zeroth-order stress carries no trace of molecular anisotropy.

	Similarly since $f^{(0)}$ is even in $\vec{w}$ and even in $\vec{\Omega}$, all odd moments vanish, i.e.
        \begin{equation}
		\label{eq:Q0}
		\vec{Q}^{(0)} = \int \!\left(\frac{1}{2}m\abs{\vec{w}}^2 + \frac{1}{2}\vec{\Omega}\cdot\mathbb{N}\vec{\Omega}\right)\vec{w}\, f^{(0)}\,d\vec{v}\,d\vec{\alpha}\,d\dot{\vec{\alpha}} = \mathbf{0}.
        \end{equation}
        
Analogously, the odd moment of the form $\int \vec{w}\otimes(\mathbb{N}\vec{\Omega})\, f^{(0)}\,d\vec{v}\,d\vec{\alpha}\,d\dot{\vec{\alpha}}$ also vanishes, so
        \begin{equation}\label{eq:M0}
	\mathbb{M}^{(0)} = -\int \vec{w}\otimes(\mathbb{N}\vec{\Omega})\, f^{(0)}\,d\vec{v}\,d\vec{\alpha}\,d\dot{\vec{\alpha}} = \mathbf{0}.
        \end{equation}
        
	Substituting $\mathbb{T}^{(0)}=-p\mathbb{I}$, $\vec{Q}^{(0)}=\mathbf{0}$, $\mathbb{M}^{(0)}=\mathbf{0}$ into the balance laws, we obtain:
    \begin{subequations}
        \label{eq:boltzmann_curtiss_chapman_enskog}
        \begin{align}
		\frac{\partial\rho}{\partial t} + \nabla_{\vec{x}}\cdot (\rho\vec{u}) &= 0, \label{eq:euler-mass}\\[4pt]
		\rho\!\left(\frac{\partial\vec{u}}{\partial t} + (\nabla_{\vec{x}}\vec{u})\vec{u}\right) + \nabla_{\vec{x}} p &= \mathbf{0}, \label{eq:euler-mom}\\[4pt]
		\rho\!\left(\frac{\partial e}{\partial t} + \vec{u}\cdot\nabla_{\vec{x}} e\right) + p\,\nabla_{\vec{x}}\cdot\vec{u} &= 0. \label{eq:euler-energy}
        \end{align}
        These are exactly the compressible Euler equations, supplemented by
        \begin{equation}\label{eq:euler-angmom}
		\rho\left[\partial_t\vec{\mu} + (\nabla_{\vec{x}}\vec{\mu})\vec{u}\right] = \mathbf{0},
        \end{equation}
        \end{subequations}
	since $\mathbb{M}^{(0)}=\mathbf{0}$ and $\mathbb{T}^{(0)}$ is symmetric, so $\mathrm{skw}(\mathbb{T}^{(0)}-{\mathbb{T}^{(0)}}^T) = \mathbf{0}$, thus the angular momentum is conserved at the macroscopic level since $\dot{\vec{\mu}} = \mathbf{0}$.
	
	Hence the classical Chapman--Enskog procedure at zeroth-order yields the compressible Euler equations, with an additional conservation law for the macroscopic angular momentum, but with the Cauchy stress tensor $\mathbb{T}$ having no anisotropic contributions.
        
}

\rev{To retain the anisotropy in the Cauchy stress tensor, Farrell, Russo \& Zerbinati \cite{farrell} work under the} assumption that local molecular alignment has emerged, so the inertia tensor can be expressed as $\mathbb{N}=\lambda (\mathbb{I}-\vec{n}\otimes \vec{n})$, where $\vec{n}:\mathbb{R}^3\to \mathbb{S}^2$ is the nematic director describing the common alignment of the molecules\footnote{In \cite{carrillo} the Boltzmann--Curtiss equation is derived with an additional mean-field force enforcing this collective behavior.}.
The fact that the inertia tensor has the structure described above is used in \cite{farrell} to observe that the internal energy $e$ can be expressed in terms of the macroscopic velocity $\vec{u}$, the density $\rho$ and the deformation of the nematic configuration $\vec{n}$, i.e.
\begin{equation}
    e=e(\theta,\nabla_{\vec{x}}\vec{n},\rho) = \frac{5}{2} R_s \theta + \frac{\lambda}{2}\rho R_s \theta \,\mathrm{tr}(\nabla_{\vec{x}}\vec{n}^T\nabla_{\vec{x}}\vec{n}).
\end{equation}
\rev{
We can thus rewrite \eqref{eq:energy_balance_le} making use of the above expression for the internal energy, which yields
\begin{align}
	\label{eq:variational}
	\Big(-\mathbb{T}-(\frac{\partial e}{\partial(\nabla\vec{\nu})})^T\nabla_{\vec{x}}\vec{\nu}\Big):\nabla_{\vec{x}}\vec{v}+\Big(-\mathbb{M}+\vec{\nu}\times(\frac{\partial e}{\partial\nabla\vec{\nu}})\Big):\nabla_{\vec{x}}\vec{\omega}= 0,
\end{align}
where we neglected the heat flux $\vec{Q}$ since it is zero at zeroth-order and the vector product operation between a vector and a tensor has been overloaded as $\vec{\nu}\times S = \varepsilon_{iqp}\nu_qS_{jp}$.  
}
In \cite{farrell} rational thermodynamic ideas are then invoked to \rev{identify the constitutive relations for the thermodynamic fluxes $\mathbb{T}$ and $\mathbb{M}$ to be
\begin{equation}
	\mathbb{T}^{(*)}=-\left(\frac{\partial e}{\partial(\nabla\vec{\nu})}\right)^T\nabla_{\vec{x}}\vec{\nu}, \qquad \mathbb{M}^{(*)}=\vec{\nu}\times\left(\frac{\partial e}{\partial\nabla\vec{\nu}}\right),\label{eq:constitutiveNollColemann}
\end{equation}	
}
This procedure yields the following governing equations for a compressible inviscid calamitic fluid,
\begin{subequations}
  \label{eq:LeslieEricksenWithPressureAndEnergy}
  \begin{align}
    \partial_t \rho + \nabla_{\vec{x}}\cdot (\rho\vec{u}) &= 0,\\
    \rho\Big[\partial_t \vec{u} + (\nabla_{\vec{x}}\vec{u})\vec{u}\Big] + \nabla_{\vec{x}}\cdot\Big(p\mathbb{I}+p\frac{\lambda}{2}\nabla_{\vec{x}}\vec{n}^T\nabla_{\vec{x}}\vec{n}\Big) &= 0,\\
    \lambda\rho {\ddot{\vec{n}}}+\nabla_{\vec{x}}\cdot\Big(p\frac{\lambda}{2}\nabla_{\vec{x}}\vec{n}\Big) &= \tau\vec{n},\\
    \rho\Big[\partial_t e +\nabla_{\vec{x}}e\cdot \vec{u}\Big]+\Big(p\mathbb{I}+p\frac{\lambda}{2}\nabla_{\vec{x}}\vec{n}^T\nabla_{\vec{x}}\vec{n}\Big):\nabla_{\vec{x}}\vec{u} &= 0,\\
    \vec{n}\cdot \vec{n} &= 1.
  \end{align} 
\end{subequations}

These equations constitute a compressible form of the inviscid Leslie--Ericksen system \cite{ericksen1962,leslie1968,stewart2004}, augmented by an energy balance and a pressure-dependent coupling to the director field.
In particular, the anisotropic contribution $\nabla_{\vec{x}}\vec{n}^T \nabla_{\vec{x}}\vec{n}$ appears explicitly in the Cauchy stress tensor.

\rev{
In both \eqref{eq:boltzmann_curtiss_chapman_enskog} and \eqref{eq:LeslieEricksenWithPressureAndEnergy}, the closure only employed the Maxwellian distribution $f^{(0)}$. Whereas the Chapman--Enskog procedure (\Cref{selec:CE}) at this order yields an isotropic system that is not compatible with experimental observations of polar fluids, applying an alternative selection procedure based on rational thermodynamics yields an anisotropic system of equations.
This demonstrates how alternative selection procedures can yield more information than the corresponding Chapman--Enskog procedure, even at the same order of truncation.
}

\section{Acknowledgements}

This work was funded by 
the Engineering and Physical Sciences Research Council [grant number EP/W026163/1],
the Science and Technology Facilities Council [grant number UKRI/ST/B000495/1],
the Donatio Universitatis Carolinae Chair ``Mathematical modelling of multicomponent systems'', 
the UKRI Digital Research Infrastructure Programme through the Science and Technology Facilities Council's Computational Science Centre for Research Communities (CoSeC), and the Swedish Research Council under grant no.~Z2021-06594 while in residence at Institut Mittag-Leffler in Djursholm, Sweden. J.~M\'{a}lek acknowledges the support of the project No.~25-16592S financed by the Czech Science Foundation (GA\v{C}R). P.~E.~Farrell, J.~M\'{a}lek and O.~Sou\v{c}ek are members of the Nečas Center for Mathematical Modeling. For the purpose of open access, the authors have applied a CC BY public copyright licence to any author accepted manuscript arising from this submission.
No new data were generated or analysed during this study.

\appendix
\section{First order Chapman--Enskog expansion}
\label{appendix_CE}
For completeness in this section we outline the derivation of the first order Chapman--Enskog expansion leading to \eqref{eq:first_order_correction}.
This derivation can be found in the literature \cite{Chapman,Harris,cercignani2}, but for the sake of clarity we report it here using the same notation as in the main text.
Starting from the Chapman--Enskog expansion and separating the scales appearing in the Boltzmann equation we obtained \eqref{eq:ce_starting}, i.e.

\begin{equation}
    \label{eq:ce_starting_2}
    f^{(1)}(\vec{x},\vec{v},t) = -\tilde{\tau} \left( \frac{\partial}{\partial t_0} f^{(0)} + \vec{v} \cdot \nabla_{\vec{x}} f^{(0)} \right) = - \tilde{\tau}D_{\vec{v}}f^{(0)},
\end{equation}
where $D_{\vec{v}}$ is the differential operator defined respectively for scalar and vector fields as
\begin{equation}
    D_{\vec{v}} \circ \coloneqq \partial_{t_0} \circ + \vec{v} \cdot \nabla_{\vec{x}} \circ \quad \forall \circ:\mathbb{R}^3\to \mathbb{R}\qquad D_{\vec{v}} \circ = \partial_{t_0} \circ + (\nabla_{\vec{x}} \circ)\vec{v} \quad \forall\circ : \mathbb{R}^3 \to \mathbb{R}^{3}.
\end{equation}
\begin{remark}
	Taking the first three moments of \eqref{eq:ce_starting_2} and using the solvability condition \eqref{eq:solvability} we obtain that the macroscopic quantities $\rho(\vec{x},t)$, $\vec{u}(\vec{x},t)$ and $e(\vec{x},t)$ satisfy the compressible Euler equations, i.e.~\eqref{euler} in combination with \eqref{eq:balance_law_mass}, \eqref{eq:balance_law_mass}, and \eqref{eq:evolution_internal_energy}, at time scale $t_0$.
\end{remark}
We observe that $f^{(0)}$ is only determined by the macroscopic quantities $\rho(\vec{x},t)$, $\vec{u}(\vec{x},t)$ and $e(\vec{x},t)$, thus we can write using the chain rule
\begin{equation}
    D_{\vec{v}} f^{(0)} = \partial_{\rho} f^{(0)} D_{\vec{v}} \rho + \partial_{\vec{u}} f^{(0)} \cdot D_{\vec{v}} \vec{u} + \partial_{e} f^{(0)} D_{\vec{v}} e.
\end{equation}
We observe that $D_{\vec{v}} \circ$ can be decomposed as $D_{\vec{v}} \circ = D_{\vec{u}} \circ +\; \vec{w} \cdot \nabla_{\vec{x}} \circ\;$ {or, $D_{\vec{v}} \circ = D_{\vec{u}} \circ +\;  (\nabla_{\vec{x}} \circ)\vec w$) for vector fields.} Furthermore we also notice that 
\begin{equation}
    D_{\vec{v}}  f^{(0)} = f^{(0)}D_{\vec{v}} \log(f^{(0)}).
\end{equation}
Computing the logarithm of the Maxwellian \eqref{eq:maxwellian} {and employing the caloric equation of state \eqref{eq:thermal_equation_of_state},} we have
\begin{equation}
    \label{eq:logmaxwellian}
    \log(f^{(0)}) = \log(\rho) - \frac{3}{2}\log(T) - \frac{\abs{\vec{w}}^2}{2T} + C,
\end{equation}
where the constant $C$ is equal to $-{\frac{3}{2}}\log(2\pi)-\log m$ and thus does not depend on the macroscopic quantities. 
We can now compute $D_{\vec{v}} \log(f^{(0)})$ as
\begin{equation}
    D_{\vec{v}} \log(f^{(0)}) = \frac{1}{\rho}D_{\vec{v}} \rho - \frac{3}{2T} D_{\vec{v}} T - \frac{1}{2T} D_{\vec{v}} \abs{\vec{w}}^2 + \frac{\abs{\vec{w}}^2}{2T^2} D_{\vec{v}} T.
\end{equation}
We now compute each material derivative separately and making use of the balance laws with constitutive relations evaluated at zeroth-order \eqref{euler}, i.e.
\begin{align}
    D_{\vec{v}} \rho &= D_{\vec{u}} \rho + \vec{w} \cdot \nabla_{\vec{x}} \rho = -\rho \nabla_{\vec{x}} \cdot \vec{u} + \vec{w} \cdot \nabla_{\vec{x}} \rho,\\
    D_{\vec{v}} T &= D_{\vec{u}} T + \vec{w} \cdot \nabla_{\vec{x}} T = -\frac{2}{3} T \nabla_{\vec{x}} \cdot \vec{u} + \vec{w} \cdot \nabla_{\vec{x}} T,\\
    D_{\vec{v}} \abs{\vec{w}}^2 &= -2\vec{w} \cdot D_{\vec{v}} \vec{u} = -2\vec{w} \cdot \left(D_{\vec{u}} \vec{u} + \vec{w} \cdot \nabla_{\vec{x}} \vec{u}\right) \nonumber \\
                                &= -2\vec{w} \cdot \left(-\frac{1}{\rho} \nabla_{\vec{x}} p + \vec{w} \cdot \nabla_{\vec{x}} \vec{u}\right)=2\frac{1}{\rho}\nabla_{\vec{x}} p \cdot \vec{w} - 2(\vec{w} \otimes \vec{w}):\nabla_{\vec{x}} \vec{u}\nonumber \\
                                &= 2\nabla_{\vec{x}} T \cdot \vec{w} + 2\frac{T}{\rho}\nabla_{\vec{x}} \rho \cdot \vec{w} - 2(\vec{w} \otimes \vec{w}):\mathbb{D}.
\end{align}
Substituting these expressions into $D_{\vec{v}} \log(f^{(0)})$ we obtain
\begin{align}
    D_{\vec{v}} \log(f^{(0)}) &= -\frac{5}{2}\frac{1}{T} \nabla_{\vec{x}} T \cdot \vec{w} - \frac{\abs{\vec{w}}^2}{3T} \nabla_{\vec{x}} \cdot \vec{u} + \frac{\abs{\vec{w}}^2}{2T^2} \nabla_{\vec{x}} T \cdot \vec{w} + \frac{1}{T} (\vec{w} \otimes \vec{w}):\mathbb{D}\nonumber\\
    &= \frac{1}{T}(\vec{w} \otimes \vec{w}) : \mathbb{D}^d + \left(\frac{\abs{\vec{w}}^2}{2 T^2} - \frac{5}{2T}\right) \vec{w} \cdot \nabla_{\vec{x}} T,
\end{align}
using the caloric equation of state \eqref{eq:thermal_equation_of_state} we finally obtain
\begin{equation}
    f^{(1)}(\vec{x},\vec{v},t) = -\tilde{\tau} f^{(0)}\left[
        \frac{3}{2}\frac{(\vec{w}\otimes \vec{w})^d}{e}:\mathbb{D}^d {+}\left(\frac{3}{4}{\frac{\abs{\vec{w}}^2}{e}} - \frac{5}{2} \right) \vec{w}\cdot \frac{\nabla_{\vec{x}} e}{e} 
    \right].
\end{equation}

As discussed in \eqref{eq:aux}, we can use this expression for $f^{(1)}$ to compute the first order contribution to the entropy production $\xi^{(1)}$ as 
\begin{align}
    \label{eq:entropyProductionFirstOrderCalc}
    \xi^{(1)}(\vec{x},t) &={\tau} \int_{\mathbb{R}^3} {k_B}\left[\frac{3}{2}\frac{(\vec{w}\otimes \vec{w})^d}{e}:\mathbb{D}^d \right]^2f^{(0)} d\vec{v} \nonumber \\
    &+2{\tau}\int_{\mathbb{R}^3} {k_B}\left[\frac{3}{2}\frac{(\vec{w}\otimes \vec{w})^d}{e}:\mathbb{D}^d \right]\cdot \left[\left(\frac{3}{4} \frac{\abs{\vec{w}}^2}{e} - {\frac{5}{2}} \right) \vec{w}\cdot \frac{\nabla_{\vec{x}} e}{e}\right] f^{(0)} d\vec{v} \nonumber \\
    &+{\tau} \int_{\mathbb{R}^3} {k_B}\left[\left(\frac{3}{4}\frac{\abs{\vec{w}}}{e}^2 - \frac{5}{2} \right) \vec{w}\cdot \frac{\nabla_{\vec{x}} e}{e}\right]^2 f^{(0)} d\vec{v} + \mathcal{O}(\mathrm{Kn}^3).
\end{align}
We begin by observing that the second term in the previous expression is zero by symmetry; since the Maxwellian is an even function of $\vec{w}$, the whole integrand is an odd function of $\vec{w}$, and in turn of $\vec{v}$.
We are left with the first and third term.
We begin computing the first term using the well known identity for multivariate Gaussian distributions, often referred to as Isserlis' theorem, i.e.
\begin{equation}
    \mathbb{E}_{\sim \mathcal{N}(0, \sigma^2\mathbb{I})}\left[(\vec{w} \otimes \vec{w}:A)(\vec{w} \otimes \vec{w}:B)\right] = \sigma^4 \left(2 A : B + \mathrm{tr}(A)\mathrm{tr}(B)\right),
\end{equation}
where $A$ and $B$ are symmetric matrices and $\mathcal{N}(0, \sigma^2\mathbb{I})$ is the multivariate Gaussian distribution with mean zero and covariance $\sigma^2\mathbb{I}$.
We will then use this identity together with the fact that the Maxwellian \eqref{eq:maxwellian} is a multivariate Gaussian distribution with mean $\vec{0}$ and variance $\frac{2}{3}e\mathbb{I}$, multiplied by the density, i.e $f^{(0)} = \rho \mathcal{N}(0, \frac{2}{3}e\mathbb{I})$. 
Thus we can compute the first term as
\begin{equation}
    {\tau} \int_{\mathbb{R}^3} {k_B}\left[\frac{3}{2}\frac{(\vec{w}\otimes \vec{w})^d}{e}:\mathbb{D}^d \right]^2f^{(0)} d\vec{v} = 2R_s\tau\rho \mathbb{D}^d : \mathbb{D}^d.
\end{equation}
Using \eqref{eq:idealgas} we end up with the following expression for the first term
\begin{equation}
    {\tau} \int_{\mathbb{R}^3}{k_B}\left[\frac{3}{2}\frac{(\vec{w}\otimes \vec{w})^d}{e}:\mathbb{D}^d \right]^2f^{(0)} d\vec{v} = 2\tau \frac{p}{\theta} \mathbb{D}^d : \mathbb{D}^d.
\end{equation}
To compute the third term in the right-hand side of \eqref{eq:entropyProductionFirstOrderCalc} we introduce the auxiliary function $g(\vec{w}) = \left(\frac{3}{4}\frac{\abs{\vec{w}}^2}{e} - \frac{5}{2} \right) \vec{w}$ , and vector $\vec{a}=\frac{\nabla_{\vec{x}}e}{e}$. \rev{Thus}, we can write the third term as
\begin{align}
    {\tau} \int_{\mathbb{R}^3} {k_B}\left[\left(\frac{3}{4}\frac{\abs{\vec{w}}^2}{e} - \frac{5}{2} \right) \vec{w}\cdot \vec{a}\right]^2 f^{(0)} d\vec{v} &= \tau R_s\rho \mathbb{E}_{\sim \mathcal{N}(0, \frac{2}{3}e\mathbb{I})}\left[(g(\vec{w})\cdot \vec{a})^2\right] 
\end{align}
Using now another well known identity for multivariate Gaussian distributions, which can be proven by rotational invariance (isotropy),
\begin{equation}
    \mathbb{E}_{\sim \mathcal{N}(0, \sigma^2 \mathbb{I})}\left[g(\vec{w})\cdot \vec{a}\right] = \frac{1}{3}\mathbb{E}_{\sim \mathcal{N}(0, \sigma^2 \mathbb{I})}\left[\abs{g(\vec{w})}^2\right]\abs{\vec{a}}^2,
\end{equation}
we can compute the third term as
\begin{equation}
    {\tau} \int_{\mathbb{R}^3} {k_B}\left[\left(\frac{3}{4}\frac{\abs{\vec{w}}^2}{e} - \frac{5}{2} \right) \vec{w}\cdot \vec{a}\right]^2 f^{(0)} d\vec{v} = \frac{5}{3}R_s\tau\rho e \abs{\vec{a}}^2.  
\end{equation}
Using the fact that $\vec{a} = \frac{\nabla_{\vec{x}} e}{e}$ is also equal to $\frac{\nabla_{\vec{x}} \theta}{\theta}$, we end up with the following expression for the third term
\begin{equation}
    {\tau} \int_{\mathbb{R}^3} {k_B}\left[\left(\frac{3}{4}\frac{\abs{\vec{w}}^2}{e} - \frac{5}{2} \right) \vec{w}\cdot \frac{\nabla_{\vec{x}} e}{e}\right]^2 f^{(0)} d\vec{v} = \frac{5}{2}R_s\tau p \frac{\abs{\nabla_{\vec{x}} \theta}^2}{\theta^{2}}.
\end{equation}
\section{Beyond the BGK approximation}
\label{sec:full_collision}

A key assumption required to perform the hybrid Chapman--Enskog--Rajagopal--Srinivasa procedure outlined in the main text is that the entropy production $\theta \xi$ is a quadratic function of the affinities $\mathbb{D}^d$ and $-\nabla_{\vec{x}} \theta/\theta$ with an inverse proportionality to $\tau$, as in \eqref{eq:tau_iota_xi}.
We have proven that this is the case in \Cref{ex:whenDoesItHold} for the BGK approximation of the Boltzmann collision operator, and we did so by proving that the viscosity and thermal conductivity scale like
\begin{equation}
    \nu \sim \frac{1}{\tau \rho}, \qquad \kappa \sim \frac{1}{\tau \rho}.
\end{equation}
The purpose of this appendix is to identify the hypotheses on the scattering kernel $B$ of the full Boltzmann collision operator \eqref{eq:full_collision}, under which the first-order Chapman--Enskog transport coefficients $\nu$ and $\kappa$ scale like
\begin{equation}
    \nu \sim \frac{1}{\tau^\alpha \rho}, \qquad \kappa \sim \frac{1}{\tau^\alpha \rho},
\end{equation}
for some positive $\alpha$ and the expression for the entropy production $\theta \xi$ remains of the form
\begin{equation}
    \label{eq:entropyProductionFullCollision}
    \theta \xi = 2\nu \mathbb{D}^d : \mathbb{D}^d + \kappa \frac{\abs{\nabla_{\vec{x}} \theta}^2}{\theta},
\end{equation}
the details this calculations can be found in \cite[Chapter 7]{ferzigerKaper}.

For the full Boltzmann operator~\eqref{eq:full_collision}, the local collision frequency, i.e.\ the reciprocal of the effective relaxation time, is defined by evaluating the scattering kernel at the local Maxwellian~$f^{(0)}$:
\begin{equation}\label{eq:collision-freq}
  \frac{1}{\tau}
  := \int_{\mathbb{R}^3}\!\int_{S^2}
       B(\beta,\varphi,\lvert v - v_*\rvert)\,
       f^{(0)}\,d\beta\,d\varphi\,dv_*.
\end{equation}
For hard spheres, $B(\beta,\varphi,\lvert v-v_*\rvert) = 4r_0\lvert v-v_*\rvert\sin\theta\cos\theta$, and a direct computation yields $1/\tau \propto \rho\sqrt{e}$, as shown in \eqref{eq:scalingT}.
For a general power-law kernel
\begin{equation}\label{eq:power-law}
  B(\beta,\varphi,\lvert v-v_*\rvert)
  = b(\beta)\,\lvert v-v_*\rvert^{\lambda},
  \qquad \lambda \geq 0,
\end{equation} the Maxwellian average of $\lvert v-v_*\rvert^\lambda$ scales as $\theta^{\lambda/2}$, by direct computations, so
\begin{equation}\label{eq:tau-powerlaw}
  \frac{1}{\tau}
  \sim \rho\,\theta^{\lambda/2} \qquad \Rightarrow \qquad \theta \sim \left(\frac{1}{\tau\rho}\right)^{2/\lambda}.
\end{equation}
For the power-law kernel~\eqref{eq:power-law}, the Chapman--Enskog procedure evaluates the transport coefficients $\nu$ and $\kappa$ via bracket integrals which still lead to
\begin{equation}\label{eq:eigenvalue-visc}
	\nu \sim \tau p = \tau \rho \mathrm{R}_s \theta \sim \mathrm{R}_s (\rho \tau)^{1 - 2/\lambda}.
\end{equation}
where now $\tau$ is the new relaxation number defined as above and $\theta\xi$ has the form of \eqref{eq:entropyProductionFullCollision} \cite[Chapter 7]{ferzigerKaper}.
An identical argument applies to $\kappa$. Hence any interaction kernel of the form \eqref{eq:power-law} for which $0\leq\lambda<2$ yields the required inverse proportionality of the transport coefficients to $\tau$.

Notice that in this case we have only identified a sufficient condition for hard interaction scattering kernels. We plan to investigate the necessary conditions for soft interactions together with the grazing limit in future work.

\color{black}
\bibliography{refs}{}
\bibliographystyle{plain}
\end{document}